\newcommand{\oomit}[1]{}

\documentclass[runningheads]{llncs}

\newif\ifdraft
\drafttrue

\ifdraft
\usepackage[draft,nompar]{commenting}
\else
\usepackage[nompar]{commenting}
\fi

\declareauthor{ts}{ts}{purple}
\authorcommand{ts}{comment}

\declareauthor{hao}{hao}{blue}
\authorcommand{hao}{comment}
\usepackage{amsfonts}

\usepackage{latexsym,amsmath,amsthm,epsfig,amssymb,graphics,amscd,amsfonts}
\usepackage{float}
\usepackage{comment}
\usepackage{verbatim}
\usepackage{color}
\usepackage[usenames,dvipsnames,table]{xcolor}
\usepackage{indentfirst}
\usepackage{latexsym,amsmath,epsfig,amssymb,graphics}
\allowdisplaybreaks
\usepackage{stmaryrd}
\usepackage{mathrsfs}
\usepackage{algorithm, algorithmic}
\usepackage{amsfonts}
\usepackage{makeidx}
\usepackage{mathtools}
\usepackage{makecell}
\usepackage{multirow}
\usepackage{siunitx}
\usepackage{caption}
\usepackage{subcaption}
\usepackage{wrapfig}
\usepackage{booktabs}

\usepackage{bbding}

\usepackage{tikz}
\usetikzlibrary{external}
\usepackage[algo2e, ruled, lined, boxed, linesnumbered, commentsnumbered]{algorithm2e}
\SetKwComment{algocomment}{$\triangleright$\ }{}

\SetCommentSty{mycommfont}
\SetKwInput{KwParam}{Hyper-parameters}
\SetKwInput{KwInput}{Input}

\usepackage[hyperfootnotes=true,colorlinks,linkcolor=RedViolet,anchorcolor=RedViolet,citecolor=blue,urlcolor=black]{hyperref}
\usepackage[capitalize]{cleveref}
\crefformat{equation}{(#2#1#3)}
\crefname{enumi}{}{}

\newtheorem{assumption}{Assumption}

\usepackage[misc]{ifsym}
\usepackage{orcidlink}
\usepackage{bbding}

\usepackage[dvipsnames]{xcolor}

\floatname{algorithm}{Algorithm}

\input{macro}

\begin{document}

\title{Solving Stochastic Constraints by Oracle-based Gradient Descent and Interval Arithmetic}


\titlerunning{Solving Stochastic Constraints with Formal Guarantees}


 \author{
 Xiakun Li\inst{1} \orcidlink{0009-0007-1663-1287}\and
 Hao Wu\inst{2} \orcidlink{0000-0001-9368-4744} \and
 Bican Xia\inst{1}  \inst{(}\Envelope\inst{)} \orcidlink{}\and
 Tengshun Yang\inst{2} \orcidlink{}\and
 Naijun Zhan\inst{3}
 }
 \authorrunning{Xiakun Li, Hao Wu, Bican Xia, Tengshun Yang and Naijun Zhan}

 \institute{
 School of Mathematical Science, Peking University, Beijing, China\and
 KLSS, Institute of Software, University of Chinese Academy of Sciences, China \and 
 Key Laboratory of High Confidence Software Technology, School of Computer Science, Peking University, Beijing, China\\
 }

\maketitle

\begin{abstract}

Stochastic constraints, which incorporate both deterministic parameters and random variables, extend classical deterministic constraints by explicitly accounting for uncertainty. These constraints are increasingly prevalent in data science, artificial intelligence, and bioinformatics; however, solving them requires addressing quantitative satisfaction problems that remain a significant challenge in computer science. In this paper, we propose a novel framework for deciding deterministic parameters that maximize the satisfaction probability. 
Our approach features a unique synergy between stochastic optimization and symbolic techniques: at the high level, it employs \emph{oracle-based stochastic gradient descent} to identify high-quality parameter candidates, while at the low level, it utilizes \emph{interval arithmetic} to compute rigorously certified lower bounds. This framework produces a sequence of sound and increasingly tight lower bounds for the true maximum satisfaction probability, supported by a high-probability convergence guarantee. We demonstrate the effectiveness and efficiency of our approach through its application to Stochastic Satisfiability Modulo Theories (SSMT) problems and a stochastic trajectory planning task.

\end{abstract}

\section{Introduction}

Constraint Satisfaction Problems (CSP) serve as a cornerstone of computer science, driving the development of a diverse array of techniques, most notably Satisfiability Modulo Theories (SMT) \cite{barrett21smt,de2008z3}, for efficient reasoning and automated solving \cite{rossi2006handbook}.
Historically, the bulk of these advancements has focused on deterministic CSPs, where all variables and constraints are fixed and certain. However, the rapid evolution of data science, artificial intelligence, and Socio-Cyber-Physical Systems has shifted the focus toward \emph{stochastic constraint problems} (SCP). These problems are not only more practically relevant in the face of real-world uncertainty but also theoretically more challenging.




In this paper, we investigate an expressive subclass of SCP wherein constraints $\varphi(\boldsymbol{x}, \boldsymbol{y})$ are given by formulas within the first-order theory of the reals. Here, $\boldsymbol{x}$ denotes a vector of deterministic parameters, while $\boldsymbol{y}$ represents random variables governed by some known probability distributions. Under this formulation, solving the stochastic constraints is cast as a quantitative optimization task: deciding the optimal parameter  $\boldsymbol{x}^*$ that maximizes the satisfaction probability of the constraint.



The problem described above can be framed as a special case of the stochastic optimization problems prevalent in supervised learning~\cite{goodfellow2016deep,shalev2014understanding}. Specifically, it corresponds to the minimization of an expected $0$-$1$ loss function, where $1$ represents constraint satisfaction and $0$ represents a violation~\cite[Sec.~2.4]{berger85book}. 
Such loss functions are highly non-convex and non-smooth, posing significant challenges for obtaining numerical solutions with probabilistic guarantees and provable convergence rates~\cite{bartlett2006convexity}. 



In this work, we introduce a novel approach to address this problem. Our framework synergistically combines \emph{oracle-based stochastic descent gradient} \cite{jin2021high}, which is used to iteratively recommend promising parameter candidates, with \emph{interval arithmetic} \cite{granvilliers2006algorithm,van1997solving}, which is used to compute lower bounds for the satisfaction probability for a given candidate. By alternatively applying these two components, our algorithm produces an increasing sequence of certified lower bounds that converges, in probability, to the true maximal satisfaction probability. 
Furthermore, we propose a new concentration inequality (see \cref{lemma:oracle_error_bound}) that removes an independence assumption in \cite{jin2021high}. This enables us to apply the analysis framework of \cite{jin2021high,Jin2024High,Jin2024Sample} to establish an improved high-probability complexity bound for our algorithm.
We demonstrate that our framework is readily applicable to a broad class of \emph{Stochastic Satisfiability Modulo Theories (SSMT)} problems over continuous domains~\cite{gao2015solving}, as well as complex \emph{trajectory planning tasks} in uncertain environments. 

\oomit{
\paragraph{Contributions.}
The main contributions of this paper are summarized as follows:
\begin{itemize}
    \item We propose a novel algorithmic framework for computing the maximum satisfaction probability of semi-algebraic stochastic constraints and the corresponding optimal parameters.
    \item 
    \item 
\end{itemize} }

\paragraph{Related Work.}
\oomit{The solution of deterministic constraints, formalized as Constraint Satisfaction Problems (CSPs) \cite{rossi2006handbook}, is a mature field with well-established modeling paradigms and powerful solvers (e.g., Z3 \cite{de2008z3}) capable of tackling large-scale problems across various domains.}

A closely related field is \emph{stochastic constraint programming}, which traditionally focuses on constraints involving decision variables and \emph{discrete} random variables. These problems are typically solved via scenario-based methods~\cite{walsh2002stochastic,campi2018introduction}. Scenario optimization is a data-driven technique that transforms a stochastic problem into a deterministic counterpart by considering a finite number of sampled realizations, or \emph{scenarios}. By satisfying the constraints across these sampled instances, one can derive formal probabilistic feasibility guarantees for the original problem. Recent extensions have further integrated \emph{confidence-based reasoning} to provide such guarantees from finite samples \cite{confidencescp}.

Another significant area is \emph{chance-constrained programming} (CCP) \cite{charnes1958,dynamic2022}, where constraints are permitted to be violated with a probability below a pre-defined threshold. However, CCP typically relies on structural convexity or specific distributional assumptions to maintain tractability. Crucially, our work diverges from these approaches by addressing constraints expressed by complex logical combinations of polynomial inequalities, rather than the single or joint inequalities standard in CCP. Our framework thus provides a more expressive symbolic-numeric alternative for scenarios with non-convex constraints and continuous uncertainty.

Within the domain of formal verification, \emph{Stochastic Satisfiability Modulo Theories (SSMT)} \cite{franzle2008stochastic,teige2008stochastic} serves as a foundational model for formalizing stochastic constraint satisfaction problems. SSMT extends classical SMT by introducing a \emph{randomized quantifier}, which enables the expression of stochastic constraints involving random variables. While initial research focused primarily on discrete domains and discrete distributions, the framework was later extended in \cite{gao2015solving} to include continuous random variables and deterministic variables over interval domains. They proposed a solving architecture combining DPLL-style search \cite{davis1962machine} with \emph{interval arithmetic} \cite{granvilliers2006algorithm}; however, their approach lacks formal convergence properties and a rigorous complexity analysis.

Our work is distinguished from the aforementioned methodologies in several key aspects. First, while existing methods for stochastic constraint programming primarily target discrete random variables, and chance-constrained programming often addresses individual or joint constraints under restrictive assumptions, our framework directly handles complex joint probabilistic constraints over continuous random variables defined by first-order polynomial formulas. 
Second, whereas prior research typically seeks a single feasible solution that satisfies a pre-defined probability threshold, our framework provides a \emph{tight certified lower bound} for the maximum satisfaction probability, representing a more general and technically demanding optimization problem that bridges the gap between probabilistic search and formal certification.



\paragraph{Paper Organization.}
Section~\ref{section 2} introduces the necessary background. Section~\ref{section 3} explains the computation of lower bounds using interval arithmetic. Section~\ref{section 4} presents our main algorithm, and a theoretical analysis of its convergence and complexity is given in Section~\ref{section 5}. Section~\ref{section 6} reports the implementation and experimental results. 
Finally, we conclude in Section~\ref{section 7}.


\section{Preliminaries}\label{section 2}

In this section, we introduce necessary preliminaries and formally define the problem of interest.

\paragraph{Basic Notations.} 
Let $\Real$ denote the set of real numbers.
We use boldface letters to denote vectors of variables, such as $\seq{x}:=(x_1,\dots,x_n)$. 
Throughout this paper, we always assume that variables are real-valued and polynomials have real coefficients.
A semialgebraic set is a subset of the real space that is defined by a polynomial formula, i.e., all atomic predicates are polynomial inequalities.
Given a function $f:\Real^n\to \Real$, we denote its gradient by $\nabla f:= (\frac{\partial f}{\partial  x_1}, \dots, \frac{\partial f}{\partial x_n})$. 
Given a constant $L>0$, a function $f:\Real^n\to \Real$ is said to be $L$-\emph{Lipschitz continuous} if $\Vert \nabla f(\boldsymbol{x}) - \nabla f(\boldsymbol{x}')\Vert \leq L \Vert  \boldsymbol{x} - \boldsymbol{x}'\Vert$ for any $\boldsymbol{x},\boldsymbol{x'}\in \Real^n$, where $\|\cdot\|$ is the standard Euclidean norm.
A function $f:\Real^n\to \Real$ is called \emph{absolutely continuous} if $\forall \varepsilon>0$, $\exists \delta >0$ such that for any finite collection of disjoint \( n \)-dimensional rectangles \( \{R_i\} \) with \( \sum_i \operatorname{vol}(R_i) < \delta \), where $\operatorname{vol}(R_i)$ is the volume of $R_i$, we have \( \sum_i \operatorname{osc}_{R_i}(f) < \varepsilon \), where $\operatorname{osc}_{R_i}(f):=\sup_{\boldsymbol{x},\boldsymbol{y}\in R_i}|f(\boldsymbol{x})-f(\boldsymbol{y})|$.

We recall some basic concepts from probability theory and recommend referring to \cite{renyi2007probability} for an in-depth treatment.
Let $\Omega$ be a \textit{sample space}, which represents all possible outcomes of a random trial. An \emph{event} is a subset of~$\Omega$. 
A pair $(\Omega,\mathcal{F})$ is said to be a \textit{measure space} if $\mathcal{F}$ is a $\sigma$-algebra of measurable events. 
A triple $(\Omega,\mathcal{F},\Prob)$ is said to be a \textit{probability space} if $\Prob$ is a probability measure satisfies that (1)~$\Prob(\emptyset) = 0$; (2)~$\Prob$ is countably additive: $\Prob(\cup_{k=1}^{\infty} A_k) = \sum_{k=1}^{\infty} \Prob(A_k),\ \forall A_1,A_2,A_3,\dots \in \mathcal{F},\ A_i \cap A_j = \emptyset\ (\forall i\neq j)$; and (3)~$\Prob(\Omega) = 1$.
Let $\mathcal{B}$ be the Borel set, which is the $\sigma$-algebra generated by all the open sets in $\Real$. 
A \emph{random variable} is a function $X:\Omega \rightarrow \Real$ such that $X^{-1}(B) \in \mathcal{F},\ \forall B \in \mathcal{B}$. 
The \textit{distribution function}
of $X$ is a function $F_X: \Real\rightarrow [0,1]$ defined by $F_X(x) := \Prob(X\leq x)$ for all $x\in \Real$.
If there exists a non-negative, integrable function $\pi:\Real \rightarrow \Real$ such that $F_X(x) = \int_{-\infty}^x \pi(t)\dd t$, then $\pi(x)$ is called the \textit{density function} of $X$, which follows $\Prob (X\in A) = \int_{A} \pi(x)\dd x$ for any $A \in \mathcal{F}$. 
The \textit{expectation} of $X$, denoted by $\Expect [X]$, is defined by the measure integral $\Expect[X] := \int_{\Real} x \pi(x) \dd x$. 
Given $A\in \mathcal{F}$,
the \textit{indicator function} of $A$ is a random variable 
defined as $\ind_{A}(\omega) := 1$ when $\omega \in A$ and $\ind_{A}(\omega) := 0$ otherwise.
A random variable $X$ is said to be sub-exponential with parameters $(\nu, b)$, where $\nu,b>0$, iff $\Expect\exp\left(\lambda(X-\Expect X)\right)\leq \exp(\lambda^2\nu^2/2),\ \forall \lambda\in [0,\frac{1}{b}]$. 




\paragraph{Problem Formulation.}
We consider a \emph{stochastic constraint} to be expressed as a polynomial formula 
\begin{align}
    \varphi(\boldsymbol{x}, \boldsymbol{y}) := \bigwedge_{i=1}^k \bigvee_{j=1}^l \left( P_{i,j} (\boldsymbol{x}, \boldsymbol{y}) \geq 0 \right),
\end{align}
where $\boldsymbol{x} \in \mathbb{R}^n$ are \emph{deterministic parameters} within the domain $D \subseteq \mathbb{R}^n$,
$\boldsymbol{y} \in \mathbb{R}^m$ are \emph{continuous random variables} subject to a known probability density function $f_{\boldsymbol{y}}(\cdot)$,
and each atomic proposition $P_{i,j} (\boldsymbol{x}, \boldsymbol{y}) \geq 0$ is a polynomial inequality. 
Without loss of generality, we assume that for any fixed $\boldsymbol{x}\in D$, each polynomial $P_{i,j}(\boldsymbol{x}, \cdot)$ has positive degree 
in $\boldsymbol{y}$; otherwise,  constraints independent of $\boldsymbol{y}$ can be absorbed into the definition of the parameter domain $D$.

For a fixed parameter $\boldsymbol{x} \in D$, the \emph{satisfaction probability} $V_\varphi(\boldsymbol{x})$ is defined as the probability that the constraint $\varphi$ holds, i.e.,
\begin{equation}\label{eq:sat-prob}
    V_\varphi(\boldsymbol{x}) := \mathbb{P}_{\boldsymbol{y}}\left( \varphi(\boldsymbol{x}, \boldsymbol{y}) \text{ is true} \right) = \int_{S(\boldsymbol{x})} f_{\boldsymbol{y}}(\boldsymbol{y}) \, d\boldsymbol{y}.
\end{equation}
where $S(\boldsymbol{x}) := \{\boldsymbol{y}\in \Real^m \mid \varphi(\boldsymbol{x}, \boldsymbol{y})\}$.

The \textbf{Maximum Satisfaction Problem} asks that,
given a stochastic constraint $\varphi(\boldsymbol{x},\boldsymbol{y})$, to find the optimal parameter $\boldsymbol{x}^* \in D$ (if exists) that maximizes the satisfaction probability $V_\varphi(\seq{x})$, i.e., to solve the following optimization problem
\begin{equation}\label{eq:max-problem}
    \max_{\boldsymbol{x} \in D} \quad V_\varphi(\boldsymbol{x}).
\end{equation}
The maximum satisfaction probability is denoted by $V_{\varphi}^*:=V_{\varphi}(\seq{x}^*)$.

Since the probability density function $f_\seq{y}$ can involve transcendental functions, finding the exact solution of the above problem is in general uncomputable~\cite{richardson69jsl-undecidable}.
Therefore, our primary goal is to compute a near-optimal parameter candidate $\seq{x}^+\in D$ as well as $V_\varphi(\seq{x}^+)$, which serves as a ``tight'' lower bound to the maximum satisfaction probability~$V_{\varphi}^*$.
To make the problem tractable, we adopt two assumptions.

\begin{assumption}\label{assump:domain}
    The domain $D\subset \Real^n$ of parameters $\seq{x}$ is compact, i.e., closed and bounded.
\end{assumption}

\begin{assumption}\label{assump:measure}
    The probability measure of random variables $\seq{y}$ is absolutely continuous with respect to (w.r.t.) the Lebesgue measure over $\mathbb{R}^m$.
    Moreover, for any axis-aligned box $I\subseteq \mathbb{R}^m$, the following probability mass is computable:
    \begin{align}
        \mathbb{P}_{\seq{y}} (\boldsymbol{y} \in I) = \int_{I} f_{\boldsymbol{y}}(\boldsymbol{y}) \, \mathrm{d}\boldsymbol{y}.
    \end{align}
\end{assumption}

\cref{assump:domain} is necessary for our convergence guarantee presented in \cref{section 5}. This is usually reasonable in many practical scenarios, and our algorithm remains sound even without this assumption.
\cref{assump:measure} ensures that $V_\varphi(\seq{x})$ in \cref{eq:sat-prob} is well-defined, and is satisfied by most common continuous distributions (e.g., uniform, normal, exponential). 

\begin{remark}\label{remark:ineq}
    Whether the constraint $\varphi(\seq{x},\seq{y})$ contain strict inequalities do not affect our algorithm's execution and formal guarantees.
    This is because, for any fixed $\boldsymbol{x}$, the set $\{\boldsymbol{y} \mid P_{i,j} (\boldsymbol{x}, \boldsymbol{y}) = 0\}$ is a subset of the zero set of a non-constant polynomial, which has Lebesgue measure zero (i.e., probability zero). 
\end{remark}

\section{Proving Lower Bounds via Interval Arithmetic}\label{section 3}

In this section, we address a very special case of \cref{eq:max-problem}:
\begin{equation*}
    \text{To compute } V_{\varphi}(\seq{x}^+) \text{ for a fixed } \seq{x}^+\in D.
\end{equation*}
By leveraging the concept of \emph{lower-bounding functions}, we demonstrate how to compute a certified and arbitrarily tight lower bound using \emph{interval arithmetic}. The procedure presented here serves as a computational primitive for our main algorithm in \cref{section 4} and provides the formal foundation for our soundness result.

\paragraph{Lower-Bounding Functions.}
The primary challenge in directly maximizing $V_\varphi(\boldsymbol{x})$ in \cref{eq:max-problem} is the integration over an implicitly defined semialgebraic region $S(\seq{x})$. To circumvent this difficulty, we employ a surrogate model of the target function $V_\varphi(\boldsymbol{x})$, characterized by the following definition.

\begin{definition}[Lower-Bounding Function]\label{def:lbf}
    Let $V_{\varphi}(\boldsymbol{x})$ be the satisfaction probability function as in \eqref{eq:sat-prob} .
    A function $B: D \to [0,1]$ is called a \emph{lower-bounding function} for $V_{\varphi}(\boldsymbol{x})$ if it satisfies:
    \begin{align}
        V_{\varphi}(\boldsymbol{x}) \geq B(\boldsymbol{x}) \quad \text{for all } \boldsymbol{x} \in D.
        \label{eq:sound_bound}
    \end{align}
\end{definition}

\begin{theorem}[Lower Bound]\label{thm:soundness}
    Let $B(\boldsymbol{x})$ be a lower-bounding function for $V_{\varphi}(\boldsymbol{x})$. 
    For any $\seq{x}\in D$, the value $B(\boldsymbol{x})$ provides a \emph{lower bound} for the maximum satisfaction probability $V_{\varphi}^*$. 
\end{theorem}

The proof of \cref{thm:soundness} follows directly from the definition of $V_{\varphi}^*$ and \cref{def:lbf}:
\begin{align}
    V_{\varphi}^* := \max_{\boldsymbol{x} \in D} V_{\varphi}(\boldsymbol{x})
    \geq \max_{\boldsymbol{x} \in D} B(\boldsymbol{x}) 
            \geq B(\boldsymbol{x}) \quad \text{for all } \boldsymbol{x} \in D.
\end{align}
Essentially, the function $B(\boldsymbol{x})$ serves as a conservative under-approximation of $V_{\varphi}(\boldsymbol{x})$. Consequently, our objective shifts to constructing a function $B(\seq{x})$ that is \emph{computationally tractable}, thereby avoiding the direct evaluation of the complex original function $V_{\varphi}(\boldsymbol{x})$.

\paragraph{Construction via Interval Arithmetic.}
We now detail the construction of $B(\seq{x})$ and the computation of $B(\seq{x}^+)$ for a fixed $\seq{x}^+ \in D$ using interval arithmetic. 
Recall from \cref{eq:sat-prob} that $V_\varphi(\seq{x}^+)$ is defined as the integral of the density function $f_\seq{y}$ over the satisfaction region $S(\seq{x}^+)$. To approximate this integral, we employ a rigorous bounding strategy:
\begin{enumerate}
    \item \textbf{Domain Truncation:} We first identify a sufficiently large bounded region $D_\seq{y}$ within the support of $\seq{y}$. For the purpose of computing a certified lower bound, the probability mass outside $D_\seq{y}$ is conservatively treated as zero.
    \item \textbf{Set Approximation:} Using standard interval arithmetic techniques, such as \emph{Interval Constraint Propagation (ICP)} and \emph{recursive subdivision}, we compute two finite collections of disjoint, axis-aligned boxes within $D_\seq{y}$:
    \begin{itemize}
        \item \emph{Inner boxes} $\{I_k\}_{k=1}^{K}$: A collection $\mathcal{I}$ where each box is proven to be contained within the satisfaction set, i.e., $\bigcup_{k=1}^K I_k \subseteq S(\seq{x}^+)$.
        \item \emph{Boundary boxes} $\{O_l\}_{l=1}^{L}$: A collection $\mathcal{O}$ where each box intersects the boundary of $S(\seq{x}^+)$ but is not contained within $S(\seq{x}^+)$ , i.e., $O_l \cap S(\boldsymbol{x}^+) \neq \emptyset$ and $O_l \not \subseteq S(\boldsymbol{x}^+)$. 
    \end{itemize}
    This partitioning yields the following formal enclosure:
    \begin{equation}
        \bigcup_{k=1}^K I_k \quad \subseteq \quad S(\seq{x}^+) \cap D_\seq{y} \quad \subseteq \quad \left( \bigcup_{k=1}^{K} I_k \cup \bigcup_{l=1}^{L} O_l \right).
    \end{equation}
    According to \cref{remark:ineq}, we can keep enlarging $D_{\seq{y}}$ and refining the boundary boxes until the volume of $\bigcup_{l=1}^{L} O_l$ is below a user-specified precision threshold $\epsilon_0 > 0$.
    \item \textbf{Mass Integration:} The lower bound $B(\seq{x}^+;\epsilon_0)$ is then obtained by summing the probability masses of the inner boxes $\mathcal{I}$, i.e.,
    \begin{equation}\label{eq:B}
        B(\boldsymbol{x}^+; \epsilon_0):= \mathbb{P}_{\boldsymbol{y}}\left( \boldsymbol{y} \in \bigcup_{k=1}^{K} I_k \right)
    \end{equation}
\end{enumerate}

Since the inner boxes are guaranteed subsets of $S(\boldsymbol{x}^+)$, the cumulative probability mass within their union provides a certified lower bound for $V_{\varphi}(\boldsymbol{x}^+)$:
\begin{equation}\label{eq:v-lower}
    B(\boldsymbol{x}^+; \epsilon_0) \leq V_{\varphi}(\boldsymbol{x}^+).
\end{equation}

Conversely, by accounting for the boundary boxes $O_l$ and the precision threshold $\epsilon_0$, an upper bound for $V_\varphi(\seq{x}^+)$ can be derived:
\begin{equation}\label{eq:v-upper}
    V_{\varphi}(\boldsymbol{x}^+) \leq B(\boldsymbol{x}^+; \epsilon_0) + \sum_{l=1}^{L} \int_{O_l} f_{\boldsymbol{y}}(\boldsymbol{y}) \, d\boldsymbol{y} \leq B(\boldsymbol{x}^+; \epsilon_0) + \epsilon_0.
\end{equation}

It follows from \cref{eq:v-lower,eq:v-upper} that interval arithmetic allows for the computation of an arbitrarily tight under-approximation of $V_{\varphi}(\seq{x}^+)$. We emphasize that, while the lower bound $B(\seq{x}^+; \epsilon_0)$ is a valid lower bound for the global maximum $V_{\varphi}^*$ (per \cref{thm:soundness}), the upper bound in \cref{eq:v-upper} applies only to the local point $\seq{x}^+$. It serves primarily as a convergence criterion to quantify the tightness of the approximation at that specific point.

The advantage of the methodology presented in this section lies in providing \emph{separation of concerns}: it decouples the \emph{search for optimal parameters} (the optimization process yielding $\seq{x}^+$) from the \emph{computation of a certified lower bound} (the evaluation of $B(\seq{x}^+; \epsilon_0)$). During the optimization phase, one may employ any heuristic or stochastic method; regardless of the search strategy used, $B(\seq{x}^+; \epsilon_0)$ remains a mathematically sound lower bound for $V_\varphi^*$. This modularity justifies our two-stage algorithmic framework, developed in the subsequent sections, which leverages the efficiency of stochastic oracles alongside the rigor of interval arithmetic.

\section{Improving Lower Bounds by Oracle-based Stochastic Gradient Descent}\label{section 4}

Having established the technique for computing a rigorous lower bound of $V_\varphi(\seq{x}^+)$ for a fixed point $\seq{x}^+ \in D$ in \cref{section 3}, we now address the primary optimization task:
\begin{equation*}
    \text{To find better } \seq{x}^+\in D \text{ with larger } V_{\varphi}(\seq{x}^+).
\end{equation*}

In \cref{subsec:0-1loss}, we first reformulate this problem into an equivalent \emph{expected 0-1 loss} minimization problem and establish the theoretical foundations for applying gradient-based methods.
In \cref{subsec:stochastic_oracles}, we introduce \emph{stochastic oracles}. 
These oracles provide probabilistic approximations of the function and gradient values via sampling, allowing for the application of a stochastic gradient descent framework  from \cite{jin2021high}.
Finally, in \cref{subsec:solving_procedure}, we present our main algorithm, which integrates oracle-based stochastic gradient descent with formal interval-based certification.



\subsection{Expected 0-1 Loss Minimization}
\label{subsec:0-1loss}
We first reformulate the optimization task \cref{eq:max-problem} into a form more amenable to stochastic search techniques. Consider the negation of $\varphi$:
\begin{equation}
    \neg \varphi(\boldsymbol{x}, \boldsymbol{y}) := \bigvee_{i=1}^k \bigwedge_{j=1}^l \left(P_{i,j} (\boldsymbol{x}, \boldsymbol{y}) < 0 \right).
\end{equation}
Recall \cref{remark:ineq} that the strict inequalities can be treated as non-strict inequalities.
We define the indicator function $\boldsymbol{1}_{\neg \varphi}: D \times \mathbb{R}^m \to \{0, 1\}$ as:
\begin{align}
    \boldsymbol{1}_{\neg \varphi}(\boldsymbol{x}, \boldsymbol{y}) =
    \begin{cases}
    1, & \text{if } \neg \varphi(\boldsymbol{x}, \boldsymbol{y}) \text{ is satisfied}, \\
    0, & \text{otherwise}.
    \end{cases}
\end{align}
For any $\boldsymbol{x} \in D$, the probability of violating the specification $\varphi$ is the expected value of this indicator function: $\mathbb{E}_{\boldsymbol{y}}[\boldsymbol{1}_{\neg \varphi}(\boldsymbol{x}, \boldsymbol{y})] = 1 - V_{\varphi}(\boldsymbol{x})$. Consequently, the original maximization problem \eqref{eq:max-problem} is equivalent to the following minimization problem:
\begin{equation}\label{eq:min-problem}
    \min_{\boldsymbol{x} \in D}\quad W_\varphi(\seq{x}) := \mathbb{E}_{\boldsymbol{y}}\left[ \boldsymbol{1}_{\neg \varphi}(\boldsymbol{x}, \boldsymbol{y}) \right].
\end{equation}
The objective $W_\varphi(\seq{x})$ is an \emph{expected $0$-$1$ loss}, representing the expected violation probability over the deterministic parameter $\boldsymbol{x}$.
Since the $0$-$1$ loss is a standard objective in stochastic optimization~\cite[Sec.~2.4]{berger85book},
this formulation directly connects our problem to a canonical stochastic optimization framework.

However, unlike standard stochastic optimization tasks, our objective function $W_\varphi(\seq{x})$ involves integration over an implicitly defined semialgebraic region $S(\seq{x})$. Consequently, the boundary between $0$ and $1$ loss is governed by some complex polynomial constraints. This raises a fundamental theoretical concern: whether $W_\varphi(\seq{x})$ possesses sufficient regularity --- specifically, the existence of its gradient --- to justify the application of stochastic gradient descent techniques. 
We address this concern with the following theorem:

\begin{theorem}\label{thm:W-cont}
    Under \cref{assump:measure}, $W_\varphi(\seq{x})$ is absolutely continuous over the domain~$D$.
\end{theorem}

The proof is technical and is deferred to \cref{app:abs_cont_proof}. Crucially, this theorem ensures that $W_\varphi(\boldsymbol{x})$ is differentiable almost everywhere over $D$, even if $D$ is unbounded (i.e., without requiring \cref{assump:domain}). 
This regularity provides the formal justification for applying gradient-based stochastic methods, as it guarantees that the search landscape is sufficiently well-behaved.

\subsection{Oracle-based Stochastic Gradient Descent}
\label{subsec:stochastic_oracles}


\paragraph{Stochastic Oracles.} 
While the absolute continuity of $W_\varphi(\seq{x})$ guarantees the existence of gradients almost everywhere, neither the objective value $W_{\varphi}(\seq{x})$ nor its gradient $\nabla W_\varphi (\seq{x})$ can be directly computed. 
To circumvent this, we adopt the framework of stochastic optimization, assuming that our algorithm has black-box access to two specific types of \emph{stochastic oracles}. These oracles provide probabilistic estimates of the objective and its first-order information, allowing the optimization process to proceed using empirical observations rather than exact values.

\begin{definition}[Stochastic Zeroth-order Oracle]
\label{def:szo}
Given a point $\seq{x}\in D$, a \emph{stochastic zeroth-order oracle} computes $\tilde{W}_\varphi(\boldsymbol{x},\boldsymbol{\xi})$ as an estimate of the function value $W_\varphi(\seq{x})$, where the function $\tilde{W}_\varphi$ is known and $\boldsymbol{\xi}$ is a vector of random variables, such that the absolute error $e(\boldsymbol{x,\xi}) := \vert \tilde{W}_\varphi(\boldsymbol{x},\boldsymbol{\xi})-W_\varphi(\boldsymbol{x})\vert$ satisfies the following conditions:
\begin{itemize}
    \item[(1)] There exists a constant $\varepsilon_W > 0$ such that $\Expect [e(\boldsymbol{x},\boldsymbol{\xi})] \leq \varepsilon_W$.
    \item[(2)] There exist constants $\nu>0,\ b>0$ such that 
    \begin{align*}
        \Expect [\exp\{\lambda(e(\boldsymbol{x},\boldsymbol{\xi})-\Expect[e(\boldsymbol{x},\boldsymbol{\xi})])\}] \leq \exp\left(\frac{\lambda^2\nu^2}{2}\right),\ \forall \lambda \in \left[0,\frac{1}{b}\right].
    \end{align*}
\end{itemize}    
\end{definition}
These conditions ensure the oracle is robust and well-behaved. The first condition guarantees that the estimate is correct on average up to a small tolerance $\varepsilon_W$. The second condition ensures that the probability of returning a wildly inaccurate result decays rapidly. Specifically, it implies that the error $e(\seq{x}, \seq{\xi})$ follows a ``one-sided'' sub-exponential distribution.
The constants $\varepsilon_W$, $\nu$, and $b$ are intrinsic to the oracle.



\begin{definition}[Stochastic First-order Oracle]
Given a point $\seq{x}\in D$ and a constant $\alpha>0$, a \emph{stochastic first-order oracle} computes $\tilde{g}(\boldsymbol{x},\boldsymbol{\xi}')$ to be an estimation of the gradient $\nabla W_\varphi(\boldsymbol{x})$, where the function $\tilde{g}$ is known and $\boldsymbol{\xi}'$ is a vector of random variables, such that the following condition holds for some constants $\varepsilon_g >0,\ \kappa >0,\ \delta \in (0,1)$:
\begin{align*}
    \Prob\left(\Vert \tilde{g}(\boldsymbol{x},\boldsymbol{\xi}')- \nabla W_\varphi(\boldsymbol{x})\Vert_2 \leq \max\{\varepsilon_g,\kappa\alpha\Vert \tilde{g}(\boldsymbol{x},\boldsymbol{\xi}')\Vert_2\}\right)\geq 1-\delta
\end{align*}    
\end{definition}

Intuitively, this condition ensures that the gradient estimate is reliable with high probability $1-\delta$. The error bound is defined by two regimes: an \emph{additive regime} controlled by the precision constant $\varepsilon_g$, and a \emph{relative regime} where the error scales with the magnitude of the estimate via $\kappa \alpha$. Essentially, this ensures the oracle provides a sufficiently accurate descent direction for the optimization algorithm.
The constants $\varepsilon_g$, $\kappa$, and $\delta$ are intrinsic to the oracle.

\paragraph{Implementing Oracles.}
To utilize the stochastic oracles defined above, we must provide concrete realizations through practical computational operations. These implementations are referred to as \emph{estimators}.
In the following, we first show how to construct the estimators, and then prove that they satisfy the requirements of stochastic oracles.

Given a candidate point $\boldsymbol{x}\in D$, we construct estimators for the function value $W_\varphi(\boldsymbol{x})$ and its gradient $\nabla W_\varphi(\boldsymbol{x})$ by drawing independent random samples from the distribution of $\seq{y}$:
\begin{itemize}
    \item The \emph{zeroth-order estimator} $\tilde{W}_\varphi$ estimates $W_\varphi(\boldsymbol{x})$ using a set $\mathscr{S}$ of i.i.d. samples $\{\boldsymbol{Y}_i\}$ drawn from the probability distribution of $\boldsymbol{y}$:
    \begin{align}
        \tilde{W}_\varphi(\boldsymbol{x}, \mathscr{S}) &:= \frac{1}{|\mathscr{S}|} \sum_{\boldsymbol{Y} \in \mathscr{S}} \boldsymbol{1}_{\neg \varphi}(\boldsymbol{x}, \boldsymbol{Y}).
        \label{eq:zeroth_order_oracle}
    \end{align}
    \item The \emph{first-order estimator} $\tilde{g}$ estimates the gradient $\nabla W_\varphi(\boldsymbol{x})$ via a \emph{smoothed finite-difference scheme}. It builds upon the zeroth-order estimator $\tilde{W}_\varphi$, a set $\mathscr{U} = \{\boldsymbol{u}_i\}$ of random direction vectors (e.g., from a standard normal distribution), and a smoothing radius $\sigma > 0$:
    \begin{align}
        \tilde{g}(\boldsymbol{x}, \mathscr{U}) &:= \frac{1}{|\mathscr{U}|} \sum_{\boldsymbol{u} \in \mathscr{U}} \frac{\tilde{W}_\varphi(\boldsymbol{x} + \sigma \boldsymbol{u}, \mathscr{S}') - \tilde{W}_\varphi(\boldsymbol{x}, \mathscr{S})}{\sigma} \, \boldsymbol{u},
        \label{eq:first_order_oracle}
    \end{align}
    where $\mathscr{S}$ and $\mathscr{S}'$ are independent sample sets for the zeroth-order estimator.
\end{itemize}

The following theorems establish that $\tilde{W}_\varphi$ and $\tilde{g}$ qualify as our desired stochastic zeroth- and first-order oracles, respectively.

\begin{theorem}\label{thm:zeroth_order_oracle}
    Let $\tilde{W}_\varphi(\boldsymbol{x}, \mathscr{S})$ be defined as in \eqref{eq:zeroth_order_oracle}, and let $N := |\mathscr{S}|$ denote the number of samples. 
    Then $\tilde{W}_\varphi(\boldsymbol{x}, \mathscr{S})$ is a stochastic zeroth-order oracle with $\varepsilon_W = 1/(2\sqrt{N})$, $\nu = b = 4e^2/\sqrt{N}$.
\end{theorem}

The proof of \cref{thm:zeroth_order_oracle} is provided in \cref{app:proof_zeroth_order}. Notably, this theorem strengthens the original result established in \cite[Prop.~1]{jin2021high} by eliminating the dependence of the parameters $\varepsilon_W$, $\nu$, and $b$ on the specific point $\seq{x}$. This uniformity across the domain $D$ is a key technical contribution, as it allows us to establish the global convergence result in  \cref{thm:single_run_convergence}. 
This is achieved by leveraging the property that  $\boldsymbol{1}_{\neg \varphi}(\boldsymbol{x}, \boldsymbol{Y}) - W_\varphi(\boldsymbol{x})$ is a centered, sub-exponential random variable. The next result for stochastic first-order oracles is directly taken from~\cite{jin2021high}.

%

\begin{theorem}[{\cite[Prop.~3]{jin2021high}}]\label{thm:first_order_oracle}
    Let $\tilde{g}(\boldsymbol{x}, \mathscr{U})$ be defined as in \cref{eq:first_order_oracle} and  $\tilde{W}_\varphi$ be a zeroth-order oracle as in \cref{def:szo}.
    Assume $\nabla W_\varphi(\boldsymbol{x})$ is $L$-Lipschitz continuous over $D$.  For a fixed $\boldsymbol{x} \in D$, constants $\alpha, \kappa, \sigma > 0$, $\delta \in (0,1)$, and defining $\varepsilon_g = 2\left(\sqrt{n}L\sigma + \sqrt{n}\varepsilon_W/\sigma\right)$, then the estimator $\tilde{g}(\boldsymbol{x}, \mathscr{U})$ is a stochastic first-order oracle when $|\mathscr{U}|$ is sufficiently large (with an explicit bound given in \cite{jin2021high}).
\end{theorem}

\paragraph{The ALOE Algorithm.}
Equipped with the stochastic oracles described above, we employ the \emph{Adaptive Learning-rate Oracle-based Evolution} (ALOE) algorithm~\cite{jin2021high} to solve the minimization problem in~\eqref{eq:min-problem}. ALOE is a stochastic gradient descent algorithm based on the two types of stochastic oracles we have defined.
A key feature of ALOE is its adaptive step-size mechanism: it dynamically adjusts the learning rate based on periodic validation checking  with the zeroth-order oracle. 
The core procedure is summarized in Algorithm \ref{alg:ALOE}.


\begin{algorithm2e}[t]
\caption{ALOE Procedure for Solving \eqref{eq:min-problem}}\label{alg:ALOE}
\SetKwInOut{Input}{Input}
\SetKwInOut{Param}{Param}
\SetKwInOut{Output}{Output}
\Input{Initial point $\boldsymbol{x}_0$, stochastic oracles $\tilde{W}_\varphi$ and $\tilde{g}$}
\Output{Candidate solution $\boldsymbol{x}^+$}
\Param{Oracle parameter $\varepsilon_W$, max step size $\alpha_{\max}>0$, initial step size $\alpha_0 \in (0, \alpha_{\max})$, max iterations $K_{\max}$, constants $\theta, \gamma \in (0,1)$}
\For{$k \gets 0$ \KwTo $K_{\max}$}{
    \algocomment{Compute gradient approximation.}
    Compute the direction $\boldsymbol{g}_k \leftarrow \tilde{g}(\boldsymbol{x}_k, \mathscr{U}_k)$ using stochastic first-order oracles\;
    propose the next point $\boldsymbol{x}_k^+ \leftarrow \boldsymbol{x}_k - \alpha_k \boldsymbol{g}_k$\;
    \algocomment{Check sufficient decrease.}
    Compute $\tilde{W}_\varphi(\boldsymbol{x}_k^+, \mathscr{S}_k^+)$, $\tilde{W}_\varphi(\boldsymbol{x}_k, \mathscr{S}_k)$ using stochastic zeroth-order oracles\;
    Check if $\boldsymbol{x}_k^+ \in D$ and a stop criterion (the Armijo condition)
    $$\tilde{W}_\varphi(\boldsymbol{x}_k^+, \mathscr{S}_k^+) \leq \tilde{W}_\varphi(\boldsymbol{x}_k, \mathscr{S}_k) - \alpha_k \theta \|\boldsymbol{g}_k\|_2^2 + 2\varepsilon_W;$$\\
    \algocomment{Accept or reject.}
    \uIf{successful}{
        Accept and increase step size: $\boldsymbol{x}_{k+1} \leftarrow \boldsymbol{x}_k^+$, $\alpha_{k+1} \leftarrow \min\{\alpha_{\max}, \gamma^{-1} \alpha_k\}$\;
    }
    \Else{
        Reject and decrease step size: $\boldsymbol{x}_{k+1} \leftarrow \boldsymbol{x}_k$, $\alpha_{k+1} \leftarrow \gamma \alpha_k$\;
    }
}
\Return $\boldsymbol{x}^+\gets \boldsymbol{x}_k$\;
\end{algorithm2e}

\subsection{Oracle-based Lower-Bounding Algorithm}
\label{subsec:solving_procedure}

Building upon the two core components developed in the previous sections: oracle-based stochastic gradient descent for identifying candidate points (\cref{subsec:stochastic_oracles}) and interval arithmetic for computing certified lower bounds (\cref{section 3}), 
we are now ready to integrate  them into a unified algorithmic procedure. This framework aims to compute a rigorous lower bound for the maximum satisfaction probability $V_\varphi^*$ defined in \cref{eq:max-problem}.


This integrated process is detailed in Algorithm \ref{alg:main_procedure}. The core strategy begins by sampling $M$ initial points within the parameter domain to ensure broad coverage of the search space. For each sampled point, we perform oracle-based stochastic gradient descent using the ALOE algorithm (\cref{alg:ALOE}), which leverages zeroth- and first-order stochastic oracles to iteratively refine the parameter candidate. Upon the termination of ALOE at a candidate point $\seq{x}^+$, we apply the interval arithmetic techniques established in \cref{section 3} to compute a mathematically rigorous lower bound $B(\seq{x}^+;\epsilon_0)$. By retaining the maximum lower bound identified across all $M$ trials, we construct a sound and increasingly tight certified under-approximation for the maximum satisfaction probability $V^*$ of the target problem.

\begin{algorithm2e}[t]
\caption{Oracle-based Lower-Bounding Algorithm for Solving \eqref{eq:max-problem}}\label{alg:main_procedure}
\SetKwInOut{Input}{Input}
\SetKwInOut{Output}{Output}
\SetKwInOut{Param}{Param}
\Input{Stochastic constraint $\varphi(\boldsymbol{x}, \boldsymbol{y})$, the parameter domain $D \subseteq \mathbb{R}^n$.
}
\Param{the maximum number of trials $M$ and precision threshold $\epsilon_0 > 0$\;}
\Output{Candidate solution $\boldsymbol{x}^+$ with a certified lower bound $l$}

\algocomment{Initialize best solution and bound found so far.}
$\boldsymbol{x}^+ \leftarrow \text{NaN}$, $l \leftarrow 0$\;
\For{$m \gets 1$ \KwTo $M$}{
    Randomly sample an initial point $\boldsymbol{x}_{0,m} \in D$\;
    \algocomment{Improve the candidate by using stochastic oracles (\ref{alg:ALOE}).}
    $\boldsymbol{x}^+_m \leftarrow \text{ALOE}(\boldsymbol{x}_{0,m})$\;
    \algocomment{Compute the lower bound via interval arithmetic (\cref{section 3})}
    $l_+ \leftarrow B(\boldsymbol{x}^+_m; \epsilon_0)$\;
    \algocomment{Record the better solution.}
    \uIf{$l_+ > l$}{
        $\boldsymbol{x}^+ \leftarrow \boldsymbol{x}^+_m$, $l \leftarrow l_+$\;
    }
}
\Return $\boldsymbol{x}^+, l$\;
\end{algorithm2e}

\section{Convergence and Complexity Analysis}\label{section 5}


In this section, we analyze the convergence rate of \cref{alg:main_procedure} and establish high-probability tail bounds on its iteration complexity (i.e., the number of iterations needed to reach a desired precision) under the following assumptions:

\begin{assumption} \label{assump:analyze}
We assume the following conditions hold for \eqref{eq:min-problem}:
\begin{enumerate}
    \item The domain $D \subset \mathbb{R}^n$ is compact with a finite diameter $\operatorname{diam}(D) = d$, i.e., $\sup\{\|\seq{x}_1-\seq{x}_2\|\}\le d$.
    \item The objective function $W_\varphi(\boldsymbol{x})$ attains its minimum $W_\varphi^* := 1 - V_\varphi^*$ at a point $\boldsymbol{x}^* \in D$.
    \item The gradient $\nabla W_\varphi(\boldsymbol{x})$ is $L$-Lipschitz continuous within a \emph{non-empty open neighborhood} $U \subset D$ of $\boldsymbol{x}^*$.
\end{enumerate}    
\end{assumption}

The first two conditions are standard in the stochastic optimization setting. The third assumption regarding local smoothness is particularly relevant for our problem: it has been shown to hold provided the specification formula $\varphi$ is non-pathological (e.g., satisfying certain regularity conditions on the boundaries of the semialgebraic sets)~\cite{dynamic2022}.


The theoretical analysis in \cite{jin2021high} relies on a relatively strong assumption, which requires that the errors of stochastic first-order oracles across iterations are independent~\cite[Assumption.~2]{jin2021high}.
However, for our optimization problem \cref{eq:min-problem}, this assumption typical does not hold.
The reason lies in that our objective function $W_\varphi:=\mathbb{E}_{\boldsymbol{y}}\left[ \boldsymbol{1}_{\neg \varphi}(\boldsymbol{x}, \boldsymbol{y}) \right]$ is induced from an indicator function of a semialgebraic set.
At the $k$th step, the absolute error (see \cref{def:szo}) depends on the current position $\seq{x}_k$, which in turn depends on the error of the previous step.

Fortunately, we show that this problem can be addressed.
The cornerstone of our technique is the following \cref{lemma:oracle_error_bound}, which provides a concentration inequality for the sum of the \emph{possibly dependent}, centered oracle error sequence, without requiring the independence assumption.
This enables us to extend the analysis framework of \cite{jin2021high} to our problem and also helps to derive improved bounds.



\begin{lemma}[Concentration of Adaptive Oracle Errors]\label{lemma:oracle_error_bound}
    Consider the zeroth-order oracle $\tilde{W}_\varphi(\boldsymbol{x}, \mathscr{S})$ defined in \eqref{eq:zeroth_order_oracle} for the $0$-$1$ loss function $W_\varphi(\boldsymbol{x}) := \mathbb{E}[\boldsymbol{1}_{\neg \varphi}(\boldsymbol{x}, \boldsymbol{y})]$. Let $\mathbf{e}_k = |\tilde{W}_\varphi(\boldsymbol{x}_k, \mathscr{S}_k) - W_\varphi(\boldsymbol{x}_k)|$ and $\mathbf{e}_k^+ = |\tilde{W}_\varphi(\boldsymbol{x}_k^+, \mathscr{S}_k^+) - W_\varphi(\boldsymbol{x}_k^+)|$ be the stochastic errors at iteration $k$ of the ALOE algorithm. Define $Z_k = \mathbf{e}_k + \mathbf{e}_k^+- \mathbb{E}[\mathbf{e}_k+\mathbf{e}_k^+]$.
    Then, for any $t > 0$ and $s > 0$, the following concentration inequality holds:
    \begin{align}
        \mathbb{P}\left( \frac{1}{t} \sum_{k=0}^{t-1} Z_k > s \right) \leq \exp\left( -\min\left\{ \frac{s^2 t}{4 \nu^2}, \frac{s t}{2 b} \right\} \right),
        \label{eq:martingale_tail_bound}
    \end{align}
    where $\nu$ and $b$ are the sub-exponential parameters as in Theorem \ref{thm:zeroth_order_oracle}.
\end{lemma}

\paragraph{Stable Set.}
We now introduce some basic concepts for our analysis.
Recall that Algorithm \ref{alg:main_procedure} outputs a candidate point $\boldsymbol{x}^+$ and a lower bound $l = B(\boldsymbol{x}^+; \epsilon_0)$. 
The \emph{optimality gap} $(V^* - l)$ consists of two parts:
\begin{equation}
    |V^* - l| = \underbrace{|V^* - V(\seq{x}^+)|}_{\text{Optimization Gap}} + \underbrace{|V(\seq{x}^+) - l|}_{\text{Verification Gap}},
\end{equation}
where the absolute value operators can be safely removed.
The verification gap is bounded by the precision $\epsilon_0$ of the interval arithmetic module in \cref{section 3}, i.e., $|V(\seq{x}^+) - l| < \epsilon_0$. The optimization gap $|V^* - V(\seq{x}^+)| = |W(\boldsymbol{x}^+) - W^*|$ is controlled by the proximity of $\boldsymbol{x}^+$ to the optimal point $\seq{x}^*$. This motivates us to analyze the probability of $\boldsymbol{x}^+$ landing in a \emph{near-optimal region}.

Given $\epsilon > 0$, we define the $\epsilon$-\emph{stable set} within the neighborhood $U\subset D$:
\begin{align}
    \mathcal{S}(\epsilon) := \{ \boldsymbol{x} \in U \mid \|\nabla W_\varphi(\boldsymbol{x})\|_2 < \epsilon \}.
    \label{eq:stable_set}
\end{align}
If $\boldsymbol{x}^+ \in \mathcal{S}(\epsilon)$, then, by applying the mean value theorem to $W$ on $U$ together with the condition $\Vert\nabla W(\boldsymbol{x})\Vert_2 < \epsilon$ that defines $\mathcal{S}(\epsilon)$, the optimization gap satisfies $|W(\boldsymbol{x}^+) - W^*| < \epsilon \cdot d$.
Consequently, the optimality gap is bounded:
\begin{align}
    \boldsymbol{x}^+ \in \mathcal{S}(\epsilon) \quad \Longrightarrow \quad |V^* - l| < \epsilon_0 + \epsilon d.
    \label{eq:error_bound_from_stability}
\end{align}
Therefore, the convergence probability $\mathbb{P}(|V^* - l| < \epsilon_0 + \epsilon d)$ is lower-bounded by $\mathbb{P}(\boldsymbol{x}^+ \in \mathcal{S}(\epsilon))$.

\paragraph{Success Probability of One Trial.}
Recall that in \cref{alg:ALOE}, the candidate $\boldsymbol{x}^+$ is produced by selecting the best result from $M$ independent runs of the ALOE algorithm, each initialized randomly in $D$. 
Let the probability of selecting an initial point within the  neighborhood $U$ be denoted by $q := \mathbb{P}(\boldsymbol{x}_0 \in U)$. 
This probability $q$ is strictly positive because $U$ is a non-empty open subset of $D$, and the initial sampling distribution (e.g., a uniform distribution over $D$) used in the estimators assigns positive measure to $U$.


We first analyze the conditional behavior of a single run starting from a ``good'' initial point.

\begin{theorem}[Conditional Single-Trial Convergence]\label{thm:single_run_convergence}
    Under Assumptions \ref{assump:domain}, \ref{assump:measure}, and \ref{assump:analyze},
    let $\epsilon > 0$ be a target gradient norm tolerance defining the stable set $\mathcal{S}(\epsilon) \subset U$. 
    There exist positive constants $C_1$, $C_2$, $\rho_1 \in (0,1)$, and $\rho_2 \in (0,1)$
    such that, for any initial point $\boldsymbol{x}_0 \in U$ and any sufficiently large number of ALOE steps $K$, the probability that the ALOE (Algorithm \ref{alg:ALOE}) output $\boldsymbol{x}^+$ lies in $\mathcal{S}(\epsilon)$ is bounded below by:
    \begin{align}
        \mathbb{P}\left( \boldsymbol{x}^+ \in \mathcal{S}(\epsilon) \mid \boldsymbol{x}_0 \in U \right) \geq 1 - C_1 \rho_1^K - C_2 \rho_2^K.
        \label{eq:single_run_prob}
    \end{align}
\end{theorem}

The proof, as well as the explicit expressions for the involved constants, is detailed in Appendix \ref{app:complexity}.
The proof adapts the high-probability complexity analysis of \cite{jin2021high} to our context. The key improvement stems from deriving a sub-exponential tail bound for the averaged oracle error $\frac{1}{t}\sum_{k} (\mathbf{e}_k + \mathbf{e}_k^+)$ without requiring the independence of individual errors $\mathbf{e}_k$, which holds specifically for the $0$-$1$ loss estimator $\tilde{W}$ (see \cref{lemma:oracle_error_bound} or refer to \cref{app:complexity} for more details).

\paragraph{Cumulative Success of Multiple Trials.}
Algorithm \ref{alg:main_procedure} performs $M$ independent ALOE trials from random initial points in $D$. The final output $\boldsymbol{x}^+$ is the candidate with the highest certified lower bound $l$, which typically corresponds to the run that entered the most promising region of the parameter space. The overall success hinges on two conditions: (1) at least one trial starts in the favorable neighborhood $U$, and (2) this particular trial then converges to the stable set $\mathcal{S}(\epsilon)$.

\begin{theorem}[High-Probability Complexity Bound]\label{thm:main_complexity}
    Under Assumptions \ref{assump:domain}, \ref{assump:measure}, and \ref{assump:analyze}, let $l$ be the final lower bound output by \cref{alg:main_procedure} after $M$ independent trials, each running ALOE for $K$ steps. 
    For any $\epsilon > 0$ and the associated $\mathcal{S}(\epsilon)$, there exist constants $C_1, C_2 > 0$, $\rho_1, \rho_2 \in (0,1)$, and a problem-dependent constant $q = \mathbb{P}(\boldsymbol{x}_0 \in U) > 0$, such that
    \begin{equation}
        \mathbb{P}\left( |V^* - l| < \epsilon_0 + \epsilon d \right)
        \geq \big(1 - \exp(- q M) \big) \cdot \big(1 - C_1 \rho_1^K - C_2 \rho_2^K \big) \label{eq:main_complexity_result} 
    \end{equation}
    where $\rho = \max(\rho_1, \rho_2)$. All constants depend solely on the problem parameters and not on $M$ or $K$.
\end{theorem}

The full proof is in Appendix \ref{app:complexity}.
The key to our proof is by noting that the ``bad'' event $|V^* - l| \geq \epsilon_0 + \epsilon d$ only occurs when at least one of the following two situations happens: (1) All $M$ independent trials start from initial points \emph{outside} the neighborhood $U$. This occurs with probability at most $(1 - q)^M \leq \exp(-q M)$. (2) Given that at least one trial starts in $U$, the \emph{best} such trial (which yields the final output $\boldsymbol{x}^+$) still fails to produce a point in $\mathcal{S}(\epsilon)$. By Theorem \ref{thm:single_run_convergence}, for a trial starting in $U$, this conditional failure probability is at most $C_1 \rho_1^K + C_2 \rho_2^K$.
Applying a union bound to the probabilities of these two failure events yields the lower bound in \eqref{eq:main_complexity_result}. 


From \ref{thm:main_complexity}, the algorithm's iteration complexity satisfies the following properties:
\begin{itemize}
    \item \textbf{Exponential Convergence in $M$:} The term $1 - \exp(- q M)$ shows that for a fixed number of ALOE steps $K$, the probability of achieving an $\epsilon$-accurate solution converges to 1 \emph{exponentially fast} as the number of independent trials $M$ increases. This quantifies the benefit of the multi-start strategy in Algorithm \ref{alg:main_procedure} for exploring the parameter space.
    \item \textbf{Exponential Convergence in $K$:} The term $1 - C_1 \rho_1^K - C_2 \rho_2^K$ indicates that for a trial that starts favorably (within $U$), its success probability converges to 1 exponentially as the number of ALOE steps $K$ increases. This reflects the local rapid convergence of the \cref{alg:ALOE}.
\end{itemize}
To summarize, this high-probability bound for the iteration complexity confirms that \cref{alg:main_procedure} is not only sound, but also efficient in probability, with convergence rates that are exponential in the main algorithmic parameters $M$ and $K$.
In practice, one can either increase the number of trials $M$ to improve the chance of a good initial exploration, or increase the number of steps $K$ of each optimization to ensure deeper exploitation from a good start.

\section{Experiments}\label{section 6}
In this section, we present the experimental results and evaluate the performance of our framework. We implemented \cref{alg:main_procedure} in Python, invoking the \textsc{Realpaver}~\cite{granvilliers2006algorithm} solver to perform the interval arithmetic required described in \cref{section 3}. 
All experiments were conducted in the WSL environment (Ubuntu 24.04 LTS) of a Windows PC, equipped with an 11th Gen Intel Core i7 processor (16 logical cores) and 16 GB of RAM.



\subsection{Stochastic Satisfiability Modulo Theories}
\label{subsec:ssmt_case}

In our first case study, we demonstrate the applicability of our framework to solving SSMT problems over continuous domains. SSMT serves as a powerful quantitative extension of classical SMT by replacing the universal quantifier $\forall$ with a \emph{randomized} quantifier $\randomized$, thereby enabling reasoning with stochastic constraints~\cite{franzle2008stochastic,papadimitriou1985games,teige2008stochastic}. 


Our method directly addresses the subclass of \emph{exists-forall} SSMT formulas:
\begin{align}
    \Phi := \exists x_1 \in D_1 \dots \exists x_n \in D_n\ \randomized_{\mu_1} y_1 \in I_1 \dots \randomized_{\mu_m} y_m \in I_m : \varphi(\boldsymbol{x}, \boldsymbol{y}).
    \label{eq:er_ssmt}
\end{align}
where the subscript $\mu_i$ indicates the probability distribution of the random variable $y_i$. 
The SSMT semantics of $\Phi$, denoted $val(\Phi)$, is defined to be the maximum satisfaction probability over the existential variables, i.e., $val(\Phi) := \max_{\boldsymbol{x} \in \boldsymbol{D}} \mathbb{P}_{\boldsymbol{y}} (\varphi(\boldsymbol{x}, \boldsymbol{y}))$~\cite{gao2015solving}.
Hence, computing $val(\Phi)$ coincides with our problem \cref{eq:max-problem}. 
Up to our knowledge, the tool implemented in \cite{gao2015solving} is the only tool for dealing with such problems, but is unfortunately no longer available. 
So here we do not compare our tool against other tools.


\paragraph{Benchmark Problems.}
We evaluate our algorithm on four SSMT benchmarks. 
For each problem, we analyze the optimal value $val(\Phi)$ and the optimal solution $\seq{x}^*$ by hand.
\begin{itemize}
    \item \textbf{Case 1 ($\Phi_1$):} A two-dimensional problem with a circular-polygonal constraint.
    \begin{align*}
        \Phi_1 := \exists x \in [-1,1]\ \randomized_{\mathcal{U}(-1,1)} y : \left( x^2 + y^2 \leq 1 \right) \wedge \left( y \geq \tfrac{1}{2} \ \vee\ y \geq \tfrac{1}{2}x + \tfrac{1}{2} \right).
    \end{align*}
    Here, $val(\Phi_1) = (\sqrt{5}-1)/4 \approx 0.3090170$ with $x^* = -\frac{\sqrt{5}}{5}\approx -0.4472136$.

    \item \textbf{Case 2 ($\Phi_2$):} A four-dimensional problem involving linear arithmetic.
    \begin{align*}
        \Phi_2 := \exists x, w \in [0,3]\ \randomized_{U(0,5)} y\ \randomized_{U(0,4)} z : \varphi(x,w,y,z),
    \end{align*}
    where $\varphi(x,w,y,z) = \left( x+y \geq z \ \vee\ w+y \geq z \right) \wedge \left( x - y \leq z \right) \wedge \left( w - y \leq z \right)$.
    Here, $val(\Phi_2) = 0.8$, achieved when $\max\{x^*, w^*\} = 2$.

    \item \textbf{Case 3 ($\Phi_3$):} The example is from \cite[Example.~2]{gao2015solving}.
    \begin{align*}
        \Phi_3 := \exists x\in [-1,1]\ \exists a\in(-\infty,+\infty)\ \exists b\in (-\infty,+\infty) \ \randomized_{\mathcal{N}(0,1)} y:\\
        \left( x^2 \leq \frac{1}{9} \lor a^3+2b \geq 0 \right) \land \left( y>0 \lor a^3+2b < -1 \right)
    \end{align*}
    Here, $val(\Phi_3) = 1$, achieved when $-\frac{1}{3}\leq x^* \leq \frac{1}{3}$ and $(a^*)^3 +2b^* <-1$.

    \item \textbf{Case 4 ($\Phi_4$):} The example is from \cite[Example.~3]{gao2015solving}.
    \begin{align*}
        \Phi_4 &:= \exists x \in [-10,10]\ \randomized_{\mathcal{U}(5,25)} y\ 
        \randomized_{\mathcal{U}(-10,10)} z : \\
        &( (x>3)\lor (y<1))\land ((z>x^2 +2) \lor (y\leq 20)) \\
        \land &((x^2>49) \lor (y>7x)) \land ((x<6) \lor (y\geq z)),
    \end{align*}
    Here, $val(\Phi_4) = \frac{23}{32}= 0.71875$, achieved when $ x^* \in (7,10]$.
\end{itemize}
 


\paragraph{Results and Analysis.}
We execute \cref{alg:main_procedure} with a budget of $M=30$ trials and $K=50$ steps per trial. 
The interval arithmetic precision, which controls the tightness of the lower bound, is set $\epsilon_0 = 0.001$, 
Other hyper-parameters can be found in \cref{app:parameters}.
The experimental results are presented in Table \ref{tab:ssmt_results}.

\begin{table}[t]
    \centering
    \caption{Experimental results for SSMT benchmarks.}\label{tab:ssmt_results}
    \begin{tabular}{ccccc}
        \toprule
        \textbf{ID}\quad & \textbf{Lower Bound}\quad & \textbf{Certification ($\boldsymbol{x}^+$)} \quad & \textbf{Error ($val(\Phi)-l$)} & \textbf{CPU Time(s)}\\
        \midrule
        $\Phi_1$ & $0.30814$ & $[-0.43897]$ & $\approx 8.7 \times 10^{-4}$ & $125.48$ \\
        $\Phi_2$ & $0.79986$ & $[1.21011,\ 1.98876]$ & $\approx 1.3 \times 10^{-4}$ & $119.97$\\
        $\Phi_3$ & $0.99999$ & $[0.09653,-0.80612,-1.85433 ]$ & $\approx 2.0 \times 10^{-9}$ &$49.09$\\
        $\Phi_4$ & $0.71874$ & $[8.71291]$ &    $\approx 4.6 \times 10^{-6}$ & $33.41$\\
        \bottomrule
    \end{tabular}
\end{table}



For $\Phi_1$ and $\Phi_2$, which satisfy the smoothness assumptions (e.g., gradient Lipschitz continuity) underlying our theoretical convergence analysis, the computed lower bounds exhibit errors on the order of $10^{-4}$ to $10^{-3}$. This aligns perfectly with the specified interval arithmetic precision $\epsilon_0 = 0.001$, empirically validating the predicted relationship between verification precision and final error.
For $\Phi_3$ and $\Phi_4$, which are adapted from \cite{gao2015solving} and involve more complex logical structure (e.g., unbounded existential domains, constraints without random variables and so on), our theoretical assumptions may not hold in strictness. Nevertheless, the algorithm demonstrated remarkable robustness, achieving exceptionally small errors ($10^{-6}$ to $10^{-9}$). This indicates that our methodology
possesses significant practical efficacy even beyond its proven theoretical guarantees.
Collectively, these results demonstrate that our framework provides reliable and high-accuracy lower bounds for continuous exists-forall SSMT problems of small scales.

\subsection{Trajectory Planning under Uncertainty}
\label{subsec:trajectory_planning}

This case study demonstrates the practical application of our framework to a stochastic trajectory planning problem. The goal is to compute a sequence of waypoints that, under environmental uncertainty, maximize the probability of satisfying collision avoidance and dynamical feasibility constraints.

\paragraph{Problem Formulation.}
We consider a 2D trajectory planning problem over $N$ time steps, originating from the initial position $(p_{x,0}, p_{y,0}) = (0,0)$. The deterministic parameters 
    $\boldsymbol{x} := (p_{x,1}, p_{y,1}, \dots, p_{x,N}, p_{y,N}) \in \mathbb{R}^{2N}$    
represents the planned coordinates, which are subject to environmental uncertainties $\seq{w} = (w_1, w_2, w_3)$. Specifically, $w_1, w_2 \sim \mathcal{N}(0, 0.5)$ represent additive Gaussian noise in the $x$ and $y$ directions (e.g., wind), while $w_3 \sim \mathcal{U}(0.8, 1.2)$ denotes a multiplicative efficiency factor (e.g., propulsion efficiency). For each step $i \in \{1, \dots, N\}$, the trajectory must satisfy \emph{safety} constraints relative to $M$ circular obstacles (with centers $\boldsymbol{o}_m$ and radius $R_m$) and \emph{reachability} constraints limited by a nominal step length $L$ (fixed as $L=3$), defined respectively as:
\begin{align}
    \phi_{\text{safe},i,m}(\boldsymbol{x}, \seq{w}) &= R_m^2 - \left\| (p_{x,i} + w_1, p_{y,i} + w_2) - \boldsymbol{o}_m \right\|_2^2 \leq 0, \\
    \phi_{\text{reach},i}(\boldsymbol{x}, \seq{w}) &= \left\| (p_{x,i}, p_{y,i}) - (p_{x,i-1}, p_{y,i-1}) \right\|_2^2 - (L \cdot w_3)^2 \leq 0.
\end{align}

The overall objective is to find the trajectory $\boldsymbol{x}^*$ that maximizes the probability of satisfying all constraints jointly, which is equivalent to minimizing the violation probability $W(\boldsymbol{x})$:
{\small \begin{align}
    \min_{\boldsymbol{x} \in \mathbb{R}^{2N}} W(\boldsymbol{x}) := \mathbb{P}_{\seq{w}}\left( \bigvee_{i=1}^{N} \bigvee_{m=1}^{M} \left( \phi_{\text{safe},i,m}(\boldsymbol{x}, \seq{w}) > 0 \right) \lor \bigvee_{i=1}^{N} \left( \phi_{\text{reach},i}(\boldsymbol{x}, \seq{w}) > 0 \right) \right). 
\end{align} } 
The satisfaction probability $V(\boldsymbol{x}) = 1 - W(\boldsymbol{x})$ called the \emph{planning confidence}.

\paragraph{Experimental Setup.}
We evaluate our algorithm on a suite of 12 benchmark instances, varying the number of steps $N$ and the number/configuration of obstacles $M$. The obstacles are randomly generated within a pre-specified range.
For each trajectory length $N \in \{4, 5, 6\}$, we generate instances with $M = 1, 2, 3, 4$ obstacles 
constructing a total of 12 distinct trajectory planning problems. 
The specific obstacle configurations for each instance are detailed in Table \ref{tab:traj_results}. 

\begin{table}[t]
\centering
\caption{Experimental results for the stochastic trajectory planning benchmarks.}\label{tab:traj_results}
\small
\begin{tabular}{ccccc}
\toprule
\textbf{ID} & \textbf{N} & \textbf{Obstacles (Center, Radius)} & \textbf{Lower Bound ($l$)} & \textbf{CPU Time (s)}  \\
\midrule
T1 & 4 & $(6,2,2.5)$                    & 0.99733 & 200.00 \\
T2 & 4 & $(2,-2,2)$, $(5,5,3)$           & 0.99672 & 628.36  \\
T3 & 4 & $(1,1,1)$, $(4,-4,2)$, $(-7,1,1.5)$ & 0.99209 & 1763.42  \\
T4 & 4 & $(2,-7,2.5)$, $(-5,-4,2)$, $(-8,2,2)$, $(6,3,1)$ & 0.99999 & 306.10 \\
T5  & 5 & $(5,5,2)$                      & 0.99417 & 499.95 \\
T6  & 5 & $(2,3,2)$, $(5,4,3)$          & 0.99295 & 589.81  \\
T7  & 5 & $(-3,3,2)$, $(6,1,2)$, $(1,-6,2)$ & 0.99929 & 457.27  \\
T8  & 5 & $(2,2,3)$, $(8,3,2)$, $(5,8,2.5)$, $(-2,8,1.5)$ & 0.99066 & 354.84 \\
T9 & 6 & $(5,5,3)$                      & 0.97648 & 381.26  \\
T10 & 6 & $(3,3,2)$, $(6,6,2.5)$         & 0.99547 & 870.75  \\
T11 & 6 & $(8,3,2)$, $(-5,-6,3)$, $(1,-2,3)$    & 0.96605 & 2089.93  \\
T12 & 6 & $(2,-8,2)$, $(-5,-3,1)$, $(-8,-3,3)$, $(4,5,1.5)$ & 0.99999 &  504.72  \\

\bottomrule
\end{tabular}
\end{table}

\paragraph{Results and Analysis.}
The algorithm parameters are consistent with Section \ref{subsec:ssmt_case},
except the interval arithmetic precision is set $\epsilon_0 = 0.1$. 
The results are summarized in Table \ref{tab:traj_results}.

The experimental results across all $12$ benchmarks validate the practical efficacy and robustness of our framework for stochastic trajectory planning. The algorithm consistently yields high-confidence plans, with certified lower bounds exceeding \textbf{0.99} in $10$ instances and remaining above $0.96$ in the others. 
The computed waypoints $\boldsymbol{x}^*$ intuitively navigate dense obstacle fields---such as in T4 and T12 where $l=0.99999$---identifying non-trivial corridors in complex geometries. 

Computational efficiency is maintained with $10$ of the $12$ instances solved under $900$ seconds; notably, the solution time \emph{does not} exhibit a direct correlation with the size of the constraints (e.g., number of obstacles $M$ or steps $N$),
as evidenced by T12 ($M=4$) solving faster than T11 ($M=3$). This aligns with the complexity analysis in Theorem \ref{thm:main_complexity}, where the cost is driven by the oracle sampling and ALOE iterations. 

In summary, the experimental results presented in this section confirms that our algorithm can efficiently solve stochastic constraints of small or middle size and produce highly accurate lower bound.
The efficiency of our algorithm hinges on two factors: (1) the convergence complexity of \cref{alg:main_procedure}, which depends on the well-behavedness of the objective function $V_\varphi$ (or $W_\varphi$); and (2) the algorithmic complexity of the interval arithmetic procedure in \cref{section 3}, which depends on the dimension of random variables. 


\section{Conclusion}\label{section 7}

In this work, we proposed a novel framework for stochastic constraint solving that bridges the gap between heuristic efficiency and formal rigor. By integrating \emph{oracle-based stochastic gradient descent} for the rapid identification of parameter candidates with \emph{interval arithmetic} for the formal certification of satisfaction probability, our approach produces sound and increasingly tight lower bounds with a probabilistic convergence guarantee. Experimental results confirm that the framework is robust and consistently yields high-confidence guarantees for simple SSMT problems and complex trajectory planning tasks.
Future work includes extending the framework to richer, more practical constraint classes (e.g., transcendental or hybrid deterministic–stochastic constraints) while preserving sound certification, and exploring its integration with classical SMT techniques.


\clearpage

\bibliographystyle{splncs04}
\bibliography{bibfiles/refer}

\newpage 
\appendix

\section{Proof of {\cref{thm:W-cont}}}
\label{app:abs_cont_proof}

\begin{proof}
    The proof proceeds in several steps, beginning with a reformulation of the problem.

    \textbf{Step 1: Problem Reformulation.}
    Recall the structure of $\neg \varphi$:
    \begin{align}
        \neg \varphi(\boldsymbol{x}, \boldsymbol{y}) = \bigvee_{i=1}^k \bigwedge_{j=1}^l \left( P_{i,j} (\boldsymbol{x}, \boldsymbol{y}) < 0 \right).
    \end{align}
    By the definition of $W_\varphi(\seq{x})$ in \cref{eq:min-problem}, we have:
    \begin{equation}\label{eq:prob_union}
        W_{\varphi}(\boldsymbol{x}) = \mathbb{P}_{\boldsymbol{y}}\left( \bigcup_{i=1}^k \bigcap_{j=1}^l \{ \boldsymbol{y} : P_{i,j}(\boldsymbol{x}, \boldsymbol{y}) < 0 \} \right) 
    \end{equation}
    We define 
    \begin{equation}\label{eq:minmax_def}
        M(\boldsymbol{x}, \boldsymbol{y}) := \min_{1 \leq i \leq k} \max_{1 \leq j \leq l} P_{i,j}(\boldsymbol{x}, \boldsymbol{y}). 
    \end{equation}
    So we have 
    \begin{equation}
        F(\boldsymbol{x}) := \mathbb{P}_{\boldsymbol{y}}( M(\boldsymbol{x}, \boldsymbol{y}) < 0 )=M(\seq{x},\seq{y}).        
    \end{equation}
    Hence, it is equivalent to show that $F(\seq{x})$ is absolutely continuous.
    
    \textbf{Step 2: Reduction via Inclusion-Exclusion.}
    From \eqref{eq:prob_union}, by applying the inclusion-exclusion principle, we can express \( F(\boldsymbol{x}) \) as a finite linear combination of terms of the form:
    \begin{equation}
        F_{I}(\boldsymbol{x}) = \mathbb{P}_{\boldsymbol{y}}\left( \bigcap_{s \in I} \{ \boldsymbol{y} : Q_s(\boldsymbol{x}, \boldsymbol{y}) < 0 \} \right),    
    \end{equation}
    where each \( Q_s \) is one of the original polynomials \( P_{i,j} \). 
    Since absolute continuity is preserved under finite linear combinations, it suffices to prove that any function of the type 
    \begin{equation}
       G(\boldsymbol{x}) := \mathbb{P}_{\boldsymbol{y}}\left( \bigcap_{s=1}^{r} \{ \boldsymbol{y} : Q_s(\boldsymbol{x}, \boldsymbol{y}) < 0 \} \right) 
    \end{equation}
    is absolutely continuous.

    \textbf{Step 3: Single-Directional Absolute Continuity.}
    Define 
    \begin{equation}
    \Phi(\boldsymbol{x}, \boldsymbol{y}) := \max_{1 \leq s \leq r} Q_s(\boldsymbol{x}, \boldsymbol{y}).    
    \end{equation}
    Recall that \( f_{\seq{y}}(\boldsymbol{y}) \) is the probability density function of \( \boldsymbol{y} \), we have
    \begin{equation}\label{eq:G}
        G(\boldsymbol{x}) = \mathbb{P}_{\boldsymbol{y}}\left( \Phi(\boldsymbol{x}, \boldsymbol{y}) < 0 \right) = \int_{\mathbb{R}^m} \boldsymbol{1}_{\{\Phi(\boldsymbol{x}, \boldsymbol{y}) < 0\}} f_{\seq{y}}(\boldsymbol{y}) \, d\boldsymbol{y}.
    \end{equation}

    Fix a direction \( \boldsymbol{e} \in \mathbb{R}^n \) and a base point \( \boldsymbol{x}_0 \). Define \( h(t) = G(\boldsymbol{x}_0 + t \boldsymbol{e}) \). In the following, we show \( h(t) \) is absolutely continuous.

    Consider its distributional derivative
    \begin{equation}
            h'(t) = -\int_{\mathbb{R}^m} f(y) \,\delta(\Phi(\boldsymbol{x}_0 + t \boldsymbol{e}, \boldsymbol{y})) \frac{\partial \Phi}{\partial t} \, dy.
    \end{equation}
    where \( \delta(\cdot) \) is the Dirac function.
     Since \( \Phi \) is piecewise polynomial and thus differentiable almost everywhere, for a fixed \( \boldsymbol{y} \), the function \( g_{\boldsymbol{y}}(t) = \Phi(\boldsymbol{x}_0 + t \boldsymbol{e}, \boldsymbol{y}) \) is a continuous, piecewise-polynomial function in \( t \) with a bounded number of pieces (depending on \( r \) in \cref{eq:G}).

    Then, since each \( Q_s \) is non-constant in \( \boldsymbol{y} \), for any fixed \( s \) and $\boldsymbol{x}_0$, and $\boldsymbol{e}$, the set of \( \boldsymbol{y} \) for which \( Q_s(\boldsymbol{x}_0 + t \boldsymbol{e}, \boldsymbol{y}) \) is identically zero as a polynomial in \( t \) has Lebesgue measure zero (similar to \cref{remark:ineq}). 
    Consequently, for almost every \( \boldsymbol{y} \), \( g_{\boldsymbol{y}}(t) \) is a non-zero piecewise polynomial, and the equation \( g_{\boldsymbol{y}}(t) = 0 \) has a uniformly bounded number of real roots \( N \) (depending only on the degrees of \( Q_s \) and \( r \)).

    This bound allows us to control the derivative of $h(t)$. Formally, we have the estimate:
    \begin{equation}
        \begin{aligned}
            \int \left| h'(t) \right| \, dt 
            &\leq \int_{\mathbb{R}^m} f_{\seq{y}}(\boldsymbol{y}) \left( \int \left| \delta(g_{\boldsymbol{y}}(t)) \, g_{\boldsymbol{y}}'(t) \right| dt \right) d\boldsymbol{y} \\
            &\leq N \int_{\mathbb{R}^m} f_{\seq{y}}(\boldsymbol{y}) \, d\boldsymbol{y}\\ 
            & = N < \infty,
        \end{aligned}
    \end{equation}
    where the inner integral counts the roots of \( g_{\boldsymbol{y}}(t)=0 \), which is bounded by \( N \) for almost every \( \boldsymbol{y} \). The finiteness of this integral implies \( h(t) \) is absolutely continuous.

    \textbf{Step 4: Multivariate Absolute Continuity.}
    The argument in Step 3 holds for every direction \( \boldsymbol{e} \), particularly for the coordinate axes. Therefore, \( G(\boldsymbol{x}) \) is absolutely continuous on almost every line parallel to a coordinate axis.
    Furthermore, its distributional partial derivatives satisfy
    \begin{equation}
    \int \left| \frac{\partial G}{\partial x_i}(\boldsymbol{x}) \right| d\boldsymbol{x} \leq N < \infty, 
    \end{equation}
    implying \( \nabla G \in L^1 \). Consequently, \( G \) belongs to the Sobolev space \( W^{1,1}_{\text{loc}}(\mathbb{R}^n) \). A fundamental result in analysis \cite{Rudin1991Functional} states that functions in \( W^{1,1}_{\text{loc}}(\mathbb{R}^n) \) are absolutely continuous in the multivariate sense.
    
    \textbf{Step 5: Conclusion.}
    Since \( G(\boldsymbol{x}) \) is absolutely continuous and \( F(\boldsymbol{x}) \) (thus also \( W_{\varphi}(\boldsymbol{x}) \)) is a finite linear combination of such functions, it follows that \( W_{\varphi}(\boldsymbol{x}) \) is absolutely continuous on \( \mathbb{R}^n \). 
\end{proof}

\section{Proof of Theorem \ref{thm:zeroth_order_oracle}}
\label{app:proof_zeroth_order}

This appendix provides the proof of Theorem \ref{thm:zeroth_order_oracle}, which establishes that $\tilde{W}$ is a probabilistic zeroth-order oracle. The core of the proof lies in the following lemma concerning the sub-exponential property of the centered indicator function.

\begin{lemma}\label{lem:centered_indicator_subexp}
    Fix a $\boldsymbol{x}\in D$ and denote $p:=W_{\varphi}(\boldsymbol{x})$. The centered random variable $Z := \boldsymbol{1}_{\neg \varphi}(\boldsymbol{x}, \boldsymbol{Y}) - p$ has variance $\operatorname{Var}(Z) \leq (1/2)^2$ and is sub-exponential with parameters $(\nu, b) = (1/2, 0)$, i.e., its moment generating function satisfies
    \begin{equation}
    \mathbb{E}[e^{\lambda Z}] \leq \exp\left( \frac{\lambda^2}{8} \right) \quad \text{for all } \lambda \in \mathbb{R}.
    \end{equation}
\end{lemma}

\begin{proof}
    For a fixed $\boldsymbol{x}$, the random variable $\boldsymbol{1}_{\neg \varphi}(\boldsymbol{x}, \boldsymbol{Y})$ follows a Bernoulli distribution with success probability $p$. Therefore, $Z$ is distributionally equivalent to $U - p$, where $U \sim \text{Bernoulli}(p)$. Its variance is $\operatorname{Var}(Z) = p(1-p) \leq 1/4 = (1/2)^2$.

    We now bound the moment generating function of $Z$.
    \begin{equation}
    \begin{aligned}
        \mathbb{E}[e^{\lambda Z}] &= e^{-\lambda p} \mathbb{E}[e^{\lambda U}] = e^{-\lambda p} \left( (1-p) + p e^{\lambda} \right) \\
        &= (1-p)e^{-\lambda p} + p e^{\lambda(1-p)}.
    \end{aligned}    
    \end{equation}
    Define $\psi(\lambda) := \log \mathbb{E}[e^{\lambda Z}]$, then we have $\psi(0) = 0$ and $\psi'(0) = \mathbb{E}[Z] = 0$.

    Consider the \emph{tilted distribution} $\mathbb{P}_\lambda$ defined by 
    \begin{equation}
       \mathbb{P}_\lambda(U=1) = \frac{p e^{\lambda}}{(1-p) + p e^{\lambda}}. 
    \end{equation}
    Let $q_\lambda = \mathbb{P}_\lambda(U=1)$. A standard property of cumulant generating functions \cite{renyi2007probability} gives $\psi''(\lambda) = \operatorname{Var}_\lambda (U) = q_\lambda (1 - q_\lambda)$. Since $q_\lambda (1 - q_\lambda) \leq 1/4$ for any $q_\lambda \in [0,1]$, we have $\psi''(\lambda) \leq 1/4$ for all $\lambda \in \mathbb{R}$.

    Applying Taylor's theorem to $\psi(\lambda)$ at $0$, there exists some $\xi\in [0,\lambda]$ such that
    \begin{equation}
    \psi(\lambda) = \underbrace{\psi(0)}_{=0} + \underbrace{\psi'(0)}_{=0} \lambda + \frac{1}{2} \psi''(\xi) \lambda^2 \leq \frac{1}{2} \cdot \frac{1}{4} \cdot \lambda^2 = \frac{\lambda^2}{8}.    
    \end{equation}
    Hence, $\mathbb{E}[e^{\lambda Z}] = e^{\psi(\lambda)} \leq \exp\left( \lambda^2 / 8 \right)$ for all $\lambda \in \mathbb{R}$. This matches the definition of a sub-exponential random variable with parameters $(\nu, b) = (1/2, 0)$, where $\nu^2/2 = 1/8$ and the tail parameter $b=0$ indicates the bound holds for all $\lambda$.
\end{proof}

\begin{proof}[Proof of Theorem \ref{thm:zeroth_order_oracle}]
    Given a sample set $\mathscr{S} = \{\boldsymbol{Y}_i\}_{i=1}^N$ of size $N = |\mathscr{S}|$, the oracle estimate is $\tilde{W}_\varphi(\boldsymbol{x}, \mathscr{S}) = \frac{1}{N} \sum_{i=1}^N \boldsymbol{1}_{\neg \varphi}(\boldsymbol{x}, \boldsymbol{Y}_i)$. The estimation error is $e(\boldsymbol{x}, \mathscr{S}) = |\tilde{W}_\varphi(\boldsymbol{x}, \mathscr{S}) - W_{\varphi}(\boldsymbol{x})|$.

    Lemma \ref{lem:centered_indicator_subexp} establishes that each term $\boldsymbol{1}_{\neg \varphi}(\boldsymbol{x}, \boldsymbol{Y}_i) - W_{\varphi}(\boldsymbol{x})$ is a centered, sub-exponential random variable with parameters $(1/2, 0)$.

   The two properties in Theorem \ref{thm:zeroth_order_oracle}, the bound on $\mathbb{E}[e(\boldsymbol{x}, \mathscr{S})]$ and the control of the moment generating function of its deviation, are standard characterizations for the concentration of such sub-exponential averages. 
   They follow directly by applying \cite[Prop.~1]{jin2021high} to our estimator. Substituting the parameters $(\nu, b) = (1/2, 0)$ into these generic results yields the specific constants $\varepsilon_W = 1/(2\sqrt{N})$ and $\nu = b = 4e^2 /\sqrt{N}$ as stated in the theorem. 
\end{proof}

\section{Proofs in \cref{section 5}}
\label{app:complexity}


\subsection{Proof of \cref{lemma:oracle_error_bound}}
\label{app:01_loss_lemma}

\begin{proof}
The proof proceeds in three main steps.

\noindent\textbf{Step 1.} First, we prove a general tail bound for a zero-mean, $(\nu_1, b_1)$ sub-exponential random variable $S$. Specifically, for any $u > 0$,
\begin{equation}
\mathbb{P}(S > u) \leq \exp\left(-\min\left\{\frac{u^2}{2\nu_1^2},\frac{u}{2b_1}\right\}\right).
\end{equation}

\noindent\emph{Proof of Step 1.} For any $\lambda \in [0, 1/b_1]$, applying Markov's inequality yields:
\begin{equation}
\begin{aligned}
\mathbb{P}(S > u) &= \mathbb{P}\bigl(\exp(\lambda S) > \exp(\lambda u)\bigr) \\
&\leq \exp(-\lambda u) \, \mathbb{E}\bigl[\exp(\lambda S)\bigr].
\end{aligned}    
\end{equation}
Since $S$ is zero-mean and $(\nu_1, b_1)$ sub-exponential, we have $\mathbb{E}[\exp(\lambda S)] \leq \exp(\lambda^2 \nu_1^2 / 2)$ for $\lambda \in [0, 1/b_1]$. Thus,
\begin{equation}
\mathbb{P}(S > u) \leq \exp\left(\frac{\lambda^2 \nu_1^2}{2} - \lambda u\right).
\end{equation}
Because $\lambda$ is arbitrary within $[0, 1/b_1]$, we minimize the right-hand side over $\lambda$:
\begin{equation}
\mathbb{P}(S > u) \leq \min_{\lambda \in [0, 1/b_1]} \exp\left(\frac{\lambda^2 \nu_1^2}{2} - \lambda u\right).
\end{equation}
The minimum of the function $f(\lambda) = \lambda^2 \nu_1^2/2 - \lambda u$ on $[0, 1/b_1]$ is attained either at the unconstrained minimizer $\lambda^* = u / \nu_1^2$ (if it lies within the interval) or at the boundary $\lambda^* = 1/b_1$. A standard calculation yields the bound:
\begin{equation}
\min_{\lambda \in [0, 1/b_1]} \exp\left(\frac{\lambda^2 \nu_1^2}{2} - \lambda u\right) \leq \exp\left(-\min\left\{\frac{u^2}{2\nu_1^2},\frac{u}{2b_1}\right\}\right).
\end{equation}
This completes the proof of Step 1.

\noindent\textbf{Step 2.} Let $S_t = \sum_{k=0}^{t-1} Z_k$ for $t \geq 1$. Clearly, $\mathbb{E}[S_t] = 0$. We now show that $S_t$ is $(\sqrt{2t}\,\nu, b)$ sub-exponential, where $\nu$ and $b$ are the sub-exponential parameters from Theorem \ref{thm:zeroth_order_oracle}, i.e.,
\begin{equation}
\mathbb{E}\bigl[\exp(\lambda S_t)\bigr] \leq \exp\bigl(t\lambda^2 \nu^2\bigr), \qquad \forall \lambda \in \left[0, \frac{1}{b}\right].
\end{equation}

\noindent\emph{Proof of Step 2.} We prove by induction on $t$. Let $\{(\boldsymbol{x}_k, \boldsymbol{x}_k^+)\}$ denote the iterates of the ALOE algorithm, which form a stochastic process. Let $\mathcal{F}_k$ be the $\sigma$-algebra generated by this process up to iteration $k$.

\noindent\textit{Preliminary estimate.} For any $k \in \mathbb{N}$ and $\lambda \in [0, 1/b]$, we condition on $\mathcal{F}_k$ and use the properties of the oracle errors. Given $\mathcal{F}_k$, the points $\boldsymbol{x}_k$ and $\boldsymbol{x}_k^+$ are fixed, and the oracle errors $\mathbf{e}_k$ and $\mathbf{e}_k^+$ become conditionally independent. Moreover, by Theorem \ref{thm:zeroth_order_oracle}, each centered error $\mathbf{e}_k - \mathbb{E}[\mathbf{e}_k]$ (and similarly for $\mathbf{e}_k^+$) is $(\nu, b)$ sub-exponential. Hence,
\begin{equation}
\begin{aligned}
&\quad \mathbb{E}\bigl[\exp(\lambda Z_k) \mid \mathcal{F}_k\bigr]\\
&= \mathbb{E}\Bigl[\exp\bigl(\lambda (\mathbf{e}_k + \mathbf{e}_k^+ - \mathbb{E}[\mathbf{e}_k+\mathbf{e}_k^+])\bigr) \mid \mathcal{F}_k\Bigr] \\
&= \mathbb{E}\bigl[\exp(\lambda (\mathbf{e}_k - \mathbb{E}[\mathbf{e}_k])) \mid \mathcal{F}_k\bigr] \,
   \mathbb{E}\bigl[\exp(\lambda (\mathbf{e}_k^+ - \mathbb{E}[\mathbf{e}_k^+])) \mid \mathcal{F}_k\bigr] \\
&\leq \exp(\lambda^2\nu^2/2) \cdot \exp(\lambda^2\nu^2/2) \\
&= \exp(\lambda^2 \nu^2), \qquad \forall \lambda \in [0, 1/b].
\end{aligned}    
\end{equation}

\noindent\textit{Base case $t=1$.} For $t=1$, $S_1 = Z_0$. Using the law of total expectation and the preliminary estimate with $k=0$,
\begin{equation}
\begin{aligned}
\mathbb{E}\bigl[\exp(\lambda S_1)\bigr]
&= \mathbb{E}\bigl[\exp(\lambda Z_0)\bigr] \\
&= \mathbb{E}_{\mathcal{F}_0}\Bigl[ \mathbb{E}\bigl[\exp(\lambda Z_0) \mid \mathcal{F}_0\bigr] \Bigr] \\
&\leq \mathbb{E}_{\mathcal{F}_0}\bigl[\exp(\lambda^2 \nu^2)\bigr] \\
&= \exp(\lambda^2 \nu^2), \qquad \forall \lambda \in [0, 1/b].
\end{aligned}    
\end{equation}

\noindent\textit{Inductive step.} Assume that for some $t \geq 1$,
\begin{equation}
\mathbb{E}\bigl[\exp(\lambda S_t)\bigr] \leq \exp\bigl(t\lambda^2 \nu^2\bigr), \qquad \forall \lambda \in [0, 1/b].
\end{equation}
Consider $S_{t+1} = Z_t + S_t$. Using the law of total expectation and the conditional independence of $Z_t$ and $S_t$ given $\mathcal{F}_t$ (since, given $\mathcal{F}_t$, the past errors and $Z_t$ depends only on oracle samples), we have:
\begin{align*}
\mathbb{E}\bigl[\exp(\lambda S_{t+1})\bigr]
&= \mathbb{E}\bigl[\exp(\lambda Z_t) \exp(\lambda S_t)\bigr] \\
&= \mathbb{E}_{\mathcal{F}_t}\Bigl[ \mathbb{E}\bigl[\exp(\lambda Z_t) \exp(\lambda S_t) \mid \mathcal{F}_t\bigr] \Bigr] \\
&= \mathbb{E}_{\mathcal{F}_t}\Bigl[ \mathbb{E}\bigl[\exp(\lambda Z_t) \mid \mathcal{F}_t\bigr] \,
   \mathbb{E}\bigl[\exp(\lambda S_t) \mid \mathcal{F}_t\bigr] \Bigr] \\
&\leq \mathbb{E}_{\mathcal{F}_t}\Bigl[ \exp(\lambda^2 \nu^2) \,
   \mathbb{E}\bigl[\exp(\lambda S_t) \mid \mathcal{F}_t\bigr] \Bigr] \quad \text{(by preliminary estimate)} \\
&= \exp(\lambda^2 \nu^2) \,
   \mathbb{E}_{\mathcal{F}_t}\Bigl[ \mathbb{E}\bigl[\exp(\lambda S_t) \mid \mathcal{F}_t\bigr] \Bigr] \\
&= \exp(\lambda^2 \nu^2) \,
   \mathbb{E}\bigl[\exp(\lambda S_t)\bigr] \quad \text{(law of total expectation again)} \\
&\leq \exp(\lambda^2 \nu^2) \cdot \exp\bigl(t\lambda^2 \nu^2\bigr) \quad \text{(by induction hypothesis)} \\
&= \exp\bigl((t+1)\lambda^2 \nu^2\bigr), \qquad \forall \lambda \in [0, 1/b].
\end{align*}
This completes the induction. Therefore, $S_t$ is a zero-mean, $(\sqrt{2t}\,\nu, b)$ sub-exponential random variable.

\noindent\textbf{Step 3.} Apply the result of Step 1 to $S_t$, with parameters $\nu_1 = \sqrt{2t}\,\nu$, $b_1 = b$, and set $u = t s > 0$. Then,
\begin{equation}
\mathbb{P}(S_t > t s) \leq \exp\left(-\min\left\{\frac{(t s)^2}{2 (\sqrt{2t}\,\nu)^2},\frac{t s}{2b}\right\}\right)
= \exp\left(-\min\left\{\frac{s^2 t}{4\nu^2},\frac{s t}{2b}\right\}\right).
\end{equation}
Consequently,
\begin{equation}
\mathbb{P}\left( \frac{1}{t} \sum_{k=0}^{t-1} Z_k > s \right)
= \mathbb{P}\left( \frac{1}{t} S_t > s \right)
= \mathbb{P}(S_t > t s)
\leq \exp\left( -\min\left\{ \frac{s^2 t}{4 \nu^2}, \frac{s t}{2 b} \right\} \right).
\end{equation}
This completes the proof of Lemma \ref{lemma:oracle_error_bound}.
\end{proof}

\subsection{Proof of Theorem \ref{thm:single_run_convergence}}
\label{app:single_run_details}


\paragraph{Explicit Form of Constants.}
We first provide a explicit form of the constants in \cref{thm:single_run_convergence}.
The constants $C_1$, $C_2$, $\rho_1$, and $\rho_2$ used in Theorem \ref{thm:single_run_convergence} originate from the analysis of the ALOE algorithm's convergence rate under stochastic oracles~\cite{jin2021high}. Their values depend on:
\begin{itemize}
    \item The sample size $N = |\mathscr{S}|$ for the zeroth-order oracle, which determines $\varepsilon_W = 1/(2\sqrt{N})$ and the sub-exponential parameters $(\nu, b)$ of the error, as given in Theorem \ref{thm:zeroth_order_oracle}.
    \item The smoothing radius $\sigma$, the gradient Lipschitz constant $L$ (within $U$), and the problem dimension $n$.
    \item The internal ALOE parameters $(\alpha_{\max}, \theta, \gamma, \kappa)$ and the first-order oracle failure parameter $\delta$.
\end{itemize}
Following the convergence theorem for ALOE \cite{jin2021high} and applying our Lemma \ref{lemma:oracle_error_bound} to handle dependent errors, we obtain the convergence rate. Specifically, the analysis in \cite[Prop.~3]{jin2021high} establishes that the number of iterations $K_{\min}$ needed to find an $\epsilon$-stationary point satisfies.
:
\begin{equation}
\mathbb{P}(K_{\min} > t) \leq \exp\left(-\frac{(p-\hat{p})^2}{2p^2}t\right) + \mathbb{P}\left(\frac{1}{t}\sum_{k=0}^{t-1} \left(\mathbf{e}_k+\mathbf{e}_k^+ - \mathbb{E}[\mathbf{e}_k+\mathbf{e}_k^+]\right) > s\right),
\end{equation}
for parameters $p, \hat{p}, s$ defined therein. 

In our setting, applying Lemma \ref{lemma:oracle_error_bound} bounds the second term by $\exp\left(-\min\left\{\frac{s^2 t}{4\nu^2}, \frac{s t}{2b}\right\}\right)$. Setting $C_1=C_2=1$, $\rho_1 = \exp\left(-\frac{(p-\hat{p})^2}{2p^2}\right)$, and $\rho_2 = \exp\left(-\min\left\{\frac{s^2}{4\nu^2}, \frac{s}{2b}\right\}\right)$ yields the simplified form $C_1 \rho_1^K + C_2 \rho_2^K$ used in our theorem statement. The existence of suitable $\hat{p}$ and $s$ ensuring $\rho_1, \rho_2 < 1$ is guaranteed by the conditions in \cite{jin2021high} when all the parameters are chosen appropriately.

\paragraph{Proof of Theorem \ref{thm:single_run_convergence}.}
The proof is based on the condition that the initial point $\boldsymbol{x}_0$ lies within the neighborhood $U$ where the gradient $\nabla W_\varphi(\boldsymbol{x})$ is $L$-Lipschitz continuous.
The core of the proof is to apply the high-probability complexity analysis of the ALOE algorithm \cite[Prop. 3]{jin2021high} to our objective function $W_\varphi(\boldsymbol{x})$, while carefully handling the dependencies introduced by our stochastic oracles.

Let $K_{\min} = \min\{ k \geq 0 : \boldsymbol{x}_k \in \mathcal{S}(\epsilon) \}$, with the convention that $K_{\min} = \infty$ if no such $k$ exists. The analysis in \cite{jin2021high} provides the following bound: for any sufficiently large $t \geq 0$,
\begin{equation}
\mathbb{P}(K_{\min} > t) \leq \exp\left(-\frac{(p-\hat{p})^2}{2p^2}t\right) + \mathbb{P}\left(\frac{1}{t}\sum_{k=0}^{t-1} Z_k > s\right),
\label{eq:jin_general_bound}
\end{equation}
where $Z_k = \mathbf{e}_k + \mathbf{e}_k^+- \mathbb{E}[\mathbf{e}_k+\mathbf{e}_k^+]$, and $p, \hat{p}, s$ are positive constants whose specific definitions (involving $\epsilon$, $L$, $\varepsilon_W$, $\varepsilon_g$, etc.) are given in \cite{jin2021high}. The first term captures the convergence rate of an idealized process, while the second term accounts for the accumulation of stochastic oracle errors.

The critical step for our problem is to bound the second term in \eqref{eq:jin_general_bound} \emph{without assuming the sequence $\{Z_k\}$ is independent}. This is where our problem-specific structure comes into play. Lemma \ref{lemma:oracle_error_bound}, which leverages the $0$-$1$ loss structure, directly provides a concentration inequality for the dependent sum:
\begin{equation}
\mathbb{P}\left(\frac{1}{t}\sum_{k=0}^{t-1} Z_k > s\right) \leq \exp\left(-\min\left\{\frac{s^2 t}{4\nu^2}, \frac{s t}{2b}\right\}\right),
\end{equation}
where $\nu$ and $b$ are the sub-exponential parameters from Theorem \ref{thm:zeroth_order_oracle}.

Substituting this bound into \eqref{eq:jin_general_bound} and setting $t = K$ (the total number of ALOE steps in our run), we obtain:
\begin{align*}
\mathbb{P}(K_{\min} > K \mid \boldsymbol{x}_0 \in U) &\leq \exp\left(-\frac{(p-\hat{p})^2}{2p^2}K\right) + \exp\left(-\min\left\{\frac{s^2 K}{4\nu^2}, \frac{s K}{2b}\right\}\right) \\
&= C_1 \rho_1^K + C_2 \rho_2^K,
\end{align*}
where we define $C_1 = C_2 = 1$, $\rho_1 = \exp\left(-\frac{(p-\hat{p})^2}{2p^2}\right)$, and $\rho_2 = \exp\left(-\min\left\{\frac{s^2}{4\nu^2}, \frac{s}{2b}\right\}\right)$. Clearly, $\rho_1, \rho_2 \in (0,1)$.

Finally, note that if $K_{\min} \leq K$, then the algorithm's output $\boldsymbol{x}^+$ (which is either $\boldsymbol{x}^+_{K_{\min}}$ or a later point) must be in $\mathcal{S}(\epsilon)$, because once the iterate enters the stable set, the gradient condition is satisfied. More formally, $\{\boldsymbol{x}^+ \notin \mathcal{S}(\epsilon)\} \subseteq \{K_{\min} > K\}$. Therefore,
\begin{equation}
\mathbb{P}\left( \boldsymbol{x}^+ \in \mathcal{S}(\epsilon) \mid \boldsymbol{x}_0 \in U \right) \geq 1 - \mathbb{P}(K_{\min} > K \mid \boldsymbol{x}_0 \in U) \geq 1 - C_1 \rho_1^K - C_2 \rho_2^K.
\end{equation}
This completes the proof of Theorem \ref{thm:single_run_convergence}. \hfill $\square$

\subsection{Proof of Theorem \ref{thm:main_complexity}}
\label{app:main_theorem_proof}
Recall that $q := \mathbb{P}(\boldsymbol{x}_0 \in U)$. 
This proof employs the conditional convergence theorem to establish the overall complexity bound.

\textit{Proof of Theorem \ref{thm:main_complexity}.}
Let $\boldsymbol{x}_{0,1}, \dots, \boldsymbol{x}_{0,M}\in D$ be the independent initial points for the $M$ trials in \cref{alg:main_procedure} according to a distribution that assigns positive measure to $U$. Let $\boldsymbol{x}^+_{m}$ be the output of the $m$-th ALOE run.

Define the events:
\begin{align*}
    A &:= \{\exists m \in \{1,\dots,M\} \text{ such that } \boldsymbol{x}_{0,m} \in U \}, \\
    B_m &:= \{\boldsymbol{x}^+_{m} \in \mathcal{S}(\epsilon) \}, \quad \text{for } m = 1,\dots,M, \\
    E &:= \{ |V^* - l| < \epsilon_0 + \epsilon d \}.
\end{align*}
Our goal is to lower-bound $\mathbb{P}(E)$. From \eqref{eq:error_bound_from_stability}, we know that if for some $m$, both $\boldsymbol{x}_{0,m} \in U$ and $B_m$ occurs, then the final certified lower bound $l$ (which is at least $B(\boldsymbol{x}^+_{m}; \epsilon_0)$) will satisfy the error bound. More formally, $E$ occurs if the event $A \cap (\bigcup_{\{m\mid \boldsymbol{x}_{0,m} \in U\}} B_m)$ occurs.

We analyze this probability by conditioning on the initial points.
\begin{align*}
\mathbb{P}(E) &\geq \mathbb{P}\left( A \cap \bigcup_{\{m\mid \boldsymbol{x}_{0,m} \in U\}} B_m \right) \\
&= \sum_{m=1}^{M} \mathbb{P}\left( \boldsymbol{x}_{0,m} \in U, B_m, \text{ and } \boldsymbol{x}_{0,j} \notin U \text{ for all } j < m \right) \\
&= \sum_{m=1}^{M} \mathbb{P}\left( \boldsymbol{x}_{0,m} \in U \right) \cdot \mathbb{P}\left( B_m \mid \boldsymbol{x}_{0,m} \in U \right) \cdot \prod_{j=1}^{m-1} \mathbb{P}\left( \boldsymbol{x}_{0,j} \notin U \right) \quad \text{(by independence)} \\
&= \sum_{m=1}^{M} q \cdot \mathbb{P}\left( B_m \mid \boldsymbol{x}_{0,m} \in U \right) \cdot (1-q)^{m-1}.
\end{align*}

Now, by Theorem \ref{thm:single_run_convergence}, for any $m$ such that $\boldsymbol{x}_{0,m} \in U$, we have $\mathbb{P}\left( B_m \mid \boldsymbol{x}_{0,m} \in U \right) \geq 1 - C_1 \rho_1^K - C_2 \rho_2^K =: P_{\text{succ}}$. Note that $P_{\text{succ}}$ is independent of $m$ and the other trials.

Therefore,
\begin{align*}
\mathbb{P}(E) &\geq \sum_{m=1}^{M} q \cdot P_{\text{succ}} \cdot (1-q)^{m-1} \\
&= P_{\text{succ}} \cdot q \sum_{m=1}^{M} (1-q)^{m-1} \\
&= P_{\text{succ}} \cdot \left(1 - (1-q)^M\right) \quad \tag{sum of finite geometric series} \\
&\geq P_{\text{succ}} \cdot \left(1 - \exp(-qM)\right) \quad \tag{using  $1-q \leq e^{-q}$ }\\
&= \left(1 - C_1 \rho_1^K - C_2 \rho_2^K\right) \cdot \left(1 - \exp(-qM)\right).
\end{align*}

This completes the rigorous proof of Theorem \ref{thm:main_complexity}. \hfill $\square$

\section{Hyper-Parameters for {\cref{alg:main_procedure}}}
\label{app:parameters}
\begin{table}[htbp]
    \centering
    \caption{Algorithm parameters for the SSMT case studies.}\label{tab:ssmt_params}
    \begin{tabular}{lc}
        \toprule
        \textbf{Parameter} & \textbf{Value} \\
        \midrule
        Zeroth-order oracle sample size ($N = |\mathscr{S}|$) & $100$ \\
        First-order oracle direction samples ($|\mathscr{U}|$) & $50$ \\
        Smoothing radius ($\sigma$) & $0.1$ \\
        Zeroth-order oracle tolerance ($\varepsilon_W$) & $0.03$ \\
        ALOE initial step size ($\alpha_0$) & $1.0$ \\
        ALOE max step size ($\alpha_{\max}$) & $4.0$ \\
        ALOE step decay ($\gamma$) & $0.8$ \\
        ALOE Armijo constant ($\theta$) & $0.2$ \\
        \bottomrule
    \end{tabular}
\end{table}

\end{document}